\documentclass[conference]{IEEEtran}
\IEEEoverridecommandlockouts
\usepackage{graphicx}
\usepackage{orcidlink}
\usepackage{textcomp}
\usepackage{algorithm}
\usepackage{algorithmicx}
\usepackage{algpseudocode}
\usepackage{cite}
\usepackage{amsmath,amssymb,amsfonts}
\allowdisplaybreaks[4]
\usepackage{makecell}
\usepackage{multirow}
\usepackage{colortbl}
\usepackage[caption=false]{subfig}
\usepackage{textcomp}
\usepackage{xcolor}
\usepackage{algpseudocode}
\usepackage{pdfpages}
\usepackage{booktabs}
\usepackage{tikz}
\usepackage{hyperref}
\usepackage{enumitem}
\usepackage{adjustbox}
\newtheorem{assumption}{Assumption}
\newtheorem{definition}{Definition}

\newtheorem{lemma}{Lemma}
\newtheorem{remark}{Remark}
\newtheorem{corollary}{Corollary}
\newtheorem{theorem}{Theorem}

\newenvironment{proof}{{\it Proof}.\ }{\hfill $\blacksquare$\par}
\newcommand{\refappendix}[1]{\hyperref[#1]{Appendix~\ref*{#1}}}

\hypersetup{colorlinks=true,linkcolor=black,citecolor=black, urlcolor=black}

\def\BibTeX{{\rm B\kern-.05em{\sc i\kern-.025em b}\kern-.08em
    T\kern-.1667em\lower.7ex\ hbox{E}\kern-.125emX}}

\begin{document}

\title{QI-DPFL: Quality-Aware and Incentive-Boosted Federated Learning with Differential Privacy \\
\thanks{\textsuperscript{*}Corresponding Author.}
\thanks{This work was supported by National Natural Science Foundation of China under Grant No. 62206320.}
}

\author{\IEEEauthorblockN{Wenhao Yuan, Xuehe Wang\textsuperscript{*}}
\IEEEauthorblockA{School of Artificial Intelligence, Sun Yat-sen University, Zhuhai, China \\
Email: \href{mailto:yuanwh7@mail2.sysu.edu.cn}{yuanwh7@mail2.sysu.edu.cn}, \href{mailto:wangxuehe@mail.sysu.edu.cn}{wangxuehe@mail.sysu.edu.cn}}
}

\maketitle

\begin{abstract}
Federated Learning (FL) has increasingly been recognized as an innovative and secure distributed model training paradigm, aiming to coordinate multiple edge clients to collaboratively train a shared model without uploading their private datasets. The challenge of encouraging mobile edge devices to participate zealously in FL model training procedures, while mitigating the privacy leakage risks during wireless transmission, remains comparatively unexplored so far. In this paper, we propose a novel approach, named \textbf{QI-DPFL} (\underline{Q}uality-Aware and \underline{I}ncentive-Boosted \underline{F}ederated \underline{L}earning with \underline{D}ifferential \underline{P}rivacy), to address the aforementioned intractable issue. To select clients with high-quality datasets, we first propose a quality-aware client selection mechanism based on the Earth Mover’s Distance (EMD) metric. Furthermore, to attract high-quality data contributors, we design an incentive-boosted mechanism that constructs the interactions between the central server and the selected clients as a two-stage Stackelberg game, where the central server designs the time-dependent reward to minimize its cost by considering the trade-off between accuracy loss and total reward allocated, and each selected client decides the privacy budget to maximize its utility. The Nash Equilibrium of the Stackelberg game is derived to find the optimal solution in each global iteration. The extensive experimental results on different real-world datasets demonstrate the effectiveness of our proposed FL framework, by realizing the goal of privacy protection and incentive compatibility.
\end{abstract}

\begin{IEEEkeywords}
Federated learning, Stackelberg game, differential privacy, client selection mechanism.
\end{IEEEkeywords}

\section{Introduction}
In the era of rapid advancements in science and technology, an unprecedented volume of data has been generated by edge devices. It is anticipated that the data volume in human society will experience geometric growth soon. Concurrently, the surge in private data is accompanied by escalating concerns over data privacy and security, drawing considerable focus from both academic and industrial sectors. Notably, the enactment of stringent data privacy regulations such as GDPR \cite{voss2016european}, poses a formidable challenge in accessing and utilizing high-quality private data for training artificial intelligence models. In addition, the huge communication costs associated with data transmission cannot be overlooked. These challenges pave the way for the emergence of innovative machine-learning technologies that mitigate the risk of privacy disclosure \cite{zhou2021truthful}.

Federated learning emerges as a compelling distributed machine-learning paradigm, offering multiple benefits including privacy preservation and communication efficiency \cite{he2023three}. However, the widespread implementation of efficient FL systems still encounters several challenges that warrant further investigation \cite{ng2020joint}. Most research concentrates on enhancing FL model performance and assumes that FL models possess adequate safety, whereas findings from \cite{zhu2019deep} indicate potential risks of significant privacy breaches in gradient propagation schemes. Differential privacy (DP) \cite{dwork2014algorithmic}, a prevalent method for safeguarding data privacy, is often employed to alleviate the adverse effects of strong noise injected while enhancing the protection level, advancements such as $\rho$-$z$CDP \cite{bun2016concentrated} and R$\Acute{\text{e}}$nyi DP \cite{mironov2017renyi} have been proposed. Recent studies have seen the integration of DP methods with FL, striving to strike a harmonious equilibrium between model performance and privacy preservation \cite{sun2022profit,wu2022adaptive}. However, most differential privacy-based FL approaches rely on standard $(\epsilon,\delta)$-DP mechanism, which is susceptible to the Catastrophe Mechanism \cite{near2021programming}. 

Additionally, an idealized assumption in current research posits that mobile devices will participate in FL model training unconditionally once invited, a notion that often falls short in practical scenarios as engaging in model training entails significant consumption of computational and communication resources, in the meanwhile, participants need to be wary of the potential risk of information leakage \cite{kang2019incentive}. Without a well-designed economic incentive mechanism, egocentric mobile devices are probably reluctant to partake in \cite{wu2021incentivizing}. Recently, incentive mechanism-based federated network optimization and FL have gradually gained extensive attention. \cite{yi2022stackelberg,xu2023personalized} focus on modeling the interactions between clients and the central server as a Stackelberg game. Besides, there has emerged a surge of auction-based FL algorithms \cite{zhou2021truthful,sun2022profit}. Contract theory-based FL models are also proposed \cite{ding2020optimal,wu2021incentivizing}. Nevertheless, most aforementioned works on incentive mechanism design overlook the security assurance during parameter transmission between the central server and edge nodes.

To incentivize the participation of mobile devices with high-quality data and eliminate the privacy threats associated with gradient disclosure, we innovatively propose a quality-aware and incentive-boosted federated learning framework based on the $\rho$-zero-concentrated differential privacy ($\rho$-$z$CDP) technique. We first design a client selection mechanism grounded in Earth Mover's Distance (EMD) metric, followed by rigorous analysis of the differentially private federated learning (DPFL) framework, which introduces artificial Gaussian noise to obscure local model parameters, thereby addressing privacy concerns. Further, based on the DPFL framework, the interactions between the heterogeneous clients and the central server are modeled as a two-stage Stackelberg game termed QI-DPFL. In Stage \uppercase\expandafter{\romannumeral1}, the central server devises a time-dependent reward for clients to jointly minimize the accuracy loss and total reward. In Stage \uppercase\expandafter{\romannumeral2}, each selected client determines the optimal privacy budget 
in accordance with the reward allocated to maximize respective utility. The multi-fold contributions of our work are summarized as follows:
\begin{itemize}[itemsep=0pt, leftmargin=*, align=right]

\item \emph{Privacy preservation and incentive mechanism in FL:} We propose a novel and efficient quality-aware and incentive-boosted federated learning framework based on $\rho$-$z$CDP mechanism, named QI-DPFL. We first select the clients with high-quality data and then model the interactions between the central server and the selected clients as a two-stage Stackelberg game, which enables each participant to freely customize its privacy budget while achieving the well model performance in the premise of protecting data privacy.

\item \emph{Earth Mover’s Distance for client selection:} We innovatively adopt the EMD metric for client selection mechanism design to sift the geographically distributed participants with high-quality datasets and improve the training performance.

\item \emph{Stackelberg Nash Equilibrium Analysis:} By analyzing the interactions between the central server and selected clients, we derive the optimal reward and the optimal privacy budget in Stage \uppercase\expandafter{\romannumeral1} and Stage  \uppercase\expandafter{\romannumeral2} respectively. Moreover, we demonstrate that the optimal strategy profile forms a Stackelberg Nash Equilibrium. Extensive experiments on different real-world datasets verify the effectiveness and security of our proposed differentially private federated learning framework.
\end{itemize}

\begin{figure*}[t]
\setlength{\abovecaptionskip}{4pt} 
\centerline{\includegraphics[width=1.0\textwidth, trim=25 0 15 0,clip]{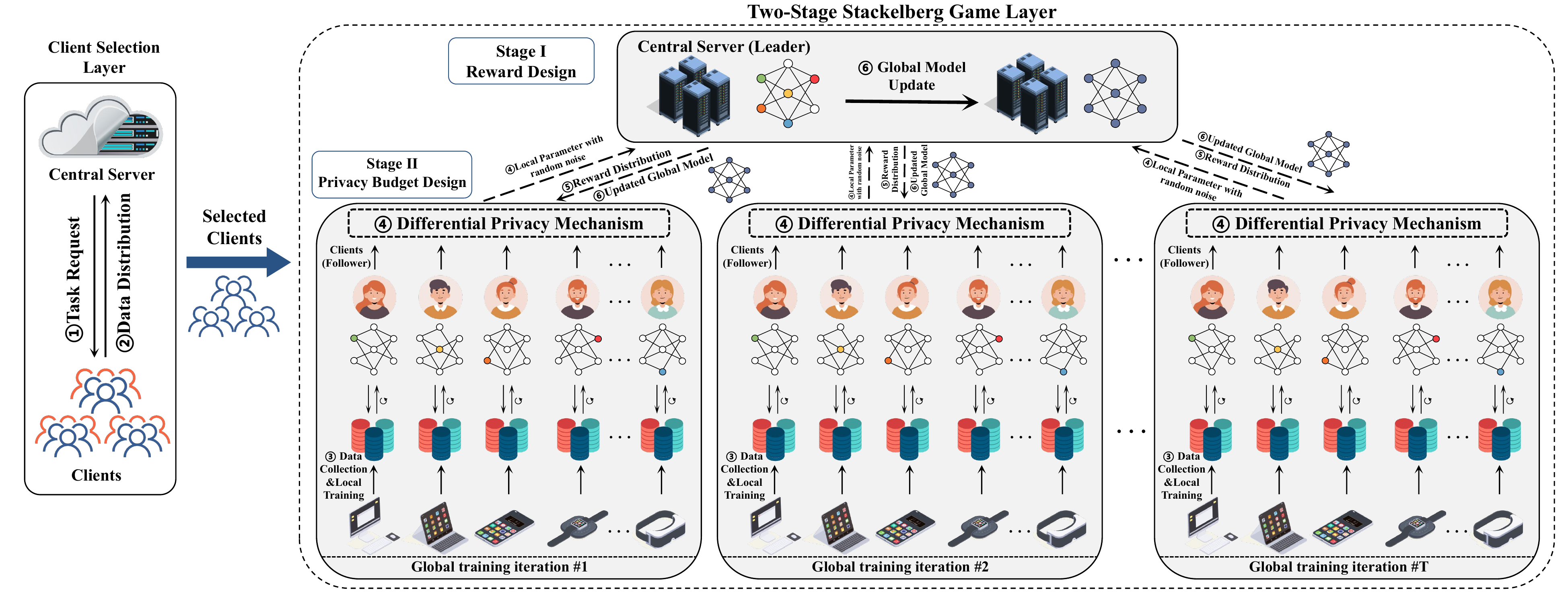}}
\vspace{-5pt}
\caption{The framework of QI-DPFL consists of the following procedures: Step 1: The central server initializes the global model and publishes the task request. Step 2: The clients submit their data distributions to the central server and are selected by EMD metric. Step 3: Each selected client collects private database $\mathcal{D}_{i}$ and performs local training. Step 4: Once finishing the local training, each client uploads the perturbed model parameter $\nabla \widetilde{F}_{i}(\boldsymbol{w_{i}(t)})$ to the central server. Step 5: The central server offers a reward $\mathcal{R}_{t}$ to all participating clients to compensate the cost of each selected participant suffered. Step 6: Upon aggregating local parameters, the global model updates and is further verified on the testing set. The training process concludes when the required accuracy or a preset number of iterations is reached. Otherwise, the central server redistributes the updated global model to clients for the next iteration (Steps 3-6).}
\label{fig1}
\vspace{-15pt}
\end{figure*}

\section{Preliminaries}

\subsection{Standard Federated Learning Model}
FL introduces an innovative decentralized machine learning paradigm in which a global model is collaboratively trained by tremendous geographically distributed clients with locally collected data. Each client engages in one or multiple epochs of mini-batch stochastic gradient descent (SGD), subsequently transmitting the updated local model to a central server for the local model aggregation and global model update. Then, the central server dispatches the updated global parameter to all clients to trigger a fresh cycle of local training until the global model converges or reaches a predefined maximum iteration.

Given $N$ clients participating in FL training and each client $\small i \in \{1, 2, \ldots, N\}$ utilizes its localized dataset $\small \mathcal{D}_{i}$ with a data size of $\small |\mathcal{D}_{i}|$ to contribute to model training. In each global iteration, each client parallelly performs $L$ ($\small L \! \geq \! 1$) epochs of SGD training to update its local parameter:
\begin{align}\label{sgd}
\boldsymbol{w_{i}^{l+1}(t)}=\boldsymbol{w_{i}^{l}(t)}-\eta_{i} \nabla F_{i}(\boldsymbol{w_{i}^{l}(t)}), 
\end{align}
where $\small \eta_{i} \! \in \! (0, 1)$ is the learning rate of client $i$, $\small l \! \in \! \{1,2, \ldots, \\ L\}$ and we define the local model $\small \boldsymbol{w_{i}^{L}(t)} \!=\! \boldsymbol{w_{i}(t)}$. The central server averages the local models from $N$ participating clients to update the global model $\boldsymbol{w(t)}$:
\begin{align} \label{local_parameter_aggregation}
\boldsymbol{w(t)}=\textstyle \sum_{i=1}^{N} \frac{|\mathcal{D}_{i}|}{\sum_{j=1}^{N} |\mathcal{D}_{j}|} \boldsymbol{w_{i}(t)}.
\end{align}

The goal of FL is to find the optimal global parameter $\boldsymbol{w}^{*}$ to minimize the global loss function $\small F(\boldsymbol{w})$:
\begin{align} \label{optimal_global_parameter}
\textstyle \boldsymbol{w}^{*}\!=\!\text{arg}\min _{\boldsymbol{w}} F(\boldsymbol{w})\!=\!\text{arg}\min _{\boldsymbol{w}} \sum_{i=1}^{N}\!\frac{D_{i}}{\sum_{j=1}^{N} D_{j}} F_{i}(\boldsymbol{w}). 
\end{align}

For the sake of theoretical analysis in the later content, we employ the following commonly used assumptions \cite{wu2021fast, sun2023stability}.
\begin{assumption} \label{smoothness_assumption}
\emph{[$\beta$-Lipschitz Smoothness] The local loss function $\small F_{i}(\boldsymbol{w})$ is $\beta$-Lipschitz smoothness for each participating client $\small i \in\{0,1, \ldots, N\}$ with any $\small \boldsymbol{w}, \boldsymbol{w}^{\prime} \in \mathbb{R}^{d}$: }
\begin{align}
\|\nabla{F_{i}(\boldsymbol{w}})-\nabla{F_{i}(\boldsymbol{w}^{\prime}})\| \leq \beta\|\boldsymbol{w}-\boldsymbol{w}^{\prime}\|. 
\end{align}
\end{assumption}
\begin{assumption} \label{convexity_assumption}
\emph{[$\lambda$-Strong Convexity] The global loss function $\small F(\boldsymbol{w})$ is $\lambda$-strong convex with any $\small \boldsymbol{w}, \boldsymbol{w}^{\prime} \in \mathbb{R}^{d}$: }
\begin{align}
F(\boldsymbol{w})\!-\!F(\boldsymbol{w}^{\prime}) \geq \langle \boldsymbol{w}\!-\!\boldsymbol{w}^{\prime}, \nabla{F(\boldsymbol{w}^{\prime})} \rangle \!+\! \dfrac{\lambda}{2} \|\boldsymbol{w}\!-\!\boldsymbol{w}^{\prime}\|^{2}. 
\end{align}
\end{assumption}

\subsection{Differential Privacy Mechanism}

For tackling the attacks such as gradient inverse attack \cite{zhu2019deep} that might disclose the original training data without accessing the datasets, $\rho$-zero-concentrated differential privacy ($\rho$-$z$CDP) was proposed in \cite{bun2016concentrated}, which attains a tight composition bound and is more suitable for analyzing the end-to-end privacy loss of iteration algorithms \cite{hu2020trading}. 

Firstly, we define a metric $\small \mathcal{L}_{p}$ to measure the privacy loss. Specifically, for a randomized mechanism $\small \mathcal{M}\! : \!\mathcal{X} \! \rightarrow \! \mathcal{E}$ with domain $\small \mathcal{X}$ and range $\small \mathcal{E}$ given any two adjacent datasets $\small \mathcal{D}, \mathcal{D}^{^{\prime}} \\ \! \subseteq \! \mathcal{X}$ with the same size but only differ by one sample, after observing an output $o \in \mathbb{R}$, the privacy loss $\small \mathcal{L}_{p}$ is given as:
\begin{align}\label{privacy_loss}
\textstyle \mathcal{L}_{p}=\text{ln}\left(\frac{Pr[\mathcal{M}(\mathcal{D})= o]}{Pr[\mathcal{M}(\mathcal{D}')= o]}\right). 
\end{align}

Then, the formal definition of $\rho$-$z$CDP is given as follows:
\begin{definition}\label{dp_definition}
\emph{A randomized mechanism $\small \mathcal{M}\!:\! \mathcal{X} \!\! \rightarrow \! \mathcal{E}$ with domain $\small \mathcal{X}$ and range $\small \mathcal{E}$ satisfies $\rho$-zero-concentrated differential privacy ($\rho$-$z$CDP) if for any $\small \alpha \!\in\! (1, \infty)$, we have:}
\begin{align}
\mathbb{E}[e^{(\alpha-1)\mathcal{L}_{p}}] \leq e^{(\alpha-1)(\rho\alpha)}. 
\end{align}
\end{definition}

Based on the Gaussian mechanism, given a query function $\small Q \!:\! \mathcal{X} \!\! \rightarrow \! \mathcal{E}$, the sensitivity of the query function $\small Q$ is defined as $\small \!\Delta{Q} \!=\! \max_{\mathcal{D},\mathcal{D}^{\prime}} \!\! \|Q(\mathcal{D})-Q(\mathcal{D}^{\prime}) \|_{2}$ for any two adjacent datasets $\small \mathcal{D}, \mathcal{D}' \! \subseteq \! \mathcal{X}$. Specifically, in the $t$-th global training iteration, by adding the artificial Gaussian noise $\small \boldsymbol{n_{i}(t)} \! \sim \! \mathcal{N}(0, \sigma_{i}^{2}(t))$, the transmitted parameter of client $i$ becomes:
\begin{align} \label{perturb_gradient}
\small \nabla \widetilde{F}_{i}(\boldsymbol{w(t)}) = \nabla F_{i}(\boldsymbol{w(t)}) + \boldsymbol{n_{i}(t)},
\end{align}
where $\frac{\Delta{Q}^{2}}{2\sigma_{i}^{2}(t)}$-$z$CDP is satisfied with $\rho_{i}^{t}\!=\!\frac{\Delta{Q}^{2}}{2\sigma_{i}^{2}(t)}$\cite{dwork2014algorithmic}. Based on the definition of the query function, we can easily derive the upper bound of the sensitivity $\small \Delta{Q}$ as given in Corollary \ref{corollary1}.

\begin{corollary} \label{corollary1}
\emph{In global iteration, by utilizing the Gaussian mechanism to perturb transmitted parameter and implementing $\rho$-$z$CDP mechanism for each participating client $i$, the sensitivity $\Delta{Q}$ of query function $Q$ is bounded by $2 C/|\mathcal{D}_{i}|$.}
\end{corollary}
\begin{proof}
For client $i$ with any two adjacent datasets $\mathcal{D}_{i}$ and $\mathcal{D}_{i}^{'}$, the sensitivity of query function $Q$ with input $\mathcal{D}_{i}$ and $\mathcal{D}_{i}^{'}$ can be obtained as follows:
\begin{align} \label{sensitivity}
\small \Delta{Q} \!=\! & \small \max_{\mathcal{D}_{i},\mathcal{D}_{i}^{'}}\|Q(\mathcal{D}_{i}) \!-\! Q(\mathcal{D}_{i}^{'}) \|_{2} \!=\! \max_{\mathcal{D}_{i},\mathcal{D}_{i}^{'}} \textstyle \frac{1}{|\mathcal{D}_{i}|} \|\! \sum_{j \in \mathcal{D}_{i}}\!\!\! \nabla f_{i}^{j}(\boldsymbol{w}, \boldsymbol{x}_{\boldsymbol{i}}^{\boldsymbol{j}}) \nonumber \\
&- \small \textstyle \sum_{j \in \mathcal{D}_{i}^{'}}\! \nabla f_{i}^{j}(\boldsymbol{w}^{\prime}, \boldsymbol{x}_{\boldsymbol{i}}^{\boldsymbol{j}})\|_{2} \!=\! \frac{2 C}{|\mathcal{D}_{i}|}, 
\end{align}
where we assume that there exists a clipping threshold $C$ for the $i$-th client's local model in $t$-th global iteration in the absence of adding artificial perturbation, i.e., $\|\boldsymbol{w_{i}(t)}\| \! \leq C$.
\end{proof}

In $t$-th global training iteration, based on Corollary \ref{corollary1} and note that $\rho_{i}^{t}=\frac{\Delta{Q}^{2}}{2\sigma_{i}^{2}(t)}$, the variance of the Gaussian random noise $\sigma_{i}^{2}(t)$ of client $i$ can be easily derived as follows:
\begin{align}\label{noise}
\textstyle \sigma_{i}^{2}(t)=\frac{2 C^{2}}{\rho_{i}^{t}|\mathcal{D}_{i}|^{2}}.
\end{align}

\section{System Model}
We propose a two-layer federated learning framework based on differential privacy technique as shown in Fig. \ref{fig1}.  In the client selection layer, the central server outsources FL tasks to all clients. The clients who are willing to participate submit their data distributions only containing the information of label frequencies to the central server, based on which, the central server selects the participants with high-quality data according to the EMD metric. Then, in the second layer, the interactions between the central server and selected clients are formulated as a two-stage Stackelberg game, where the central server decides the optimal rewards to the selected clients to incentivize them to contribute data with a high privacy budget, and each selected client determines the privacy budget $\rho_{i}^{t}$ based on the rewards and adds artificial noise to their local parameter to avoid severe privacy leakage during gradient uploading. 

\subsection{Quality-Aware Client Selection Mechanism}
From the perspective of the central server, to minimize the cost while reaching the accuracy threshold, it is crucial to select the clients with superior data quality by a metric to quantify clients' potential contributions to the FL system. To disclose the pertinent information about the local data while ensuring privacy preservation, we focus on the critical attribute of local data (i.e., data distribution) \cite{jiao2020toward}.

In FL, the data distribution varies owing to the distinct preferences of heterogeneous clients, leading to a non-independent and identically distributed (Non-IID) setting. The characteristic of Non-IID data dominantly affects the model performance, such as training accuracy \cite{mcmahan2017communication}. Inspired by \cite{jiao2020toward}, the accuracy attenuation is significantly affected by the weight divergence, which can be quantified by the Earth Mover's Distance (EMD) metric. The larger EMD value indicates higher weight divergence, thus damaging the global model quality.

In the first layer, we assume that there are a total of $H$ clients who are willing to participate in the FL model training process. Considering a $P$ class classification task that defines over a compact space $\small \mathcal{Y}$ and a label space $\small \mathcal{Z}$. The data sample $\small \mathcal{D}_{h}\!=\!\{\boldsymbol{x_{h}},y_{h}\}$ of client $\small h \! \in \! \{1, 2, \ldots, H\}$ with $\small \boldsymbol{x_{h}}\in\mathcal{Y}$ and $\small y_{h}\in\mathcal{Z}$ follows the distribution $\small \mathbb{P}_{h}$. Under the premise of the actual distribution $\small \mathbb{P}_{a}$ for the whole population, we denote the EMD of $\small \mathcal{D}_{h}$ by $\small \theta_{h}$, which can be calculated as follows:
\begin{align}\label{emd}
\small \theta_{h}=\textstyle \sum_{j \in \mathcal{Y}}\left\|\mathbb{P}_{h}(y=j)-\mathbb{P}_{a}(y=j)\right\|,
\end{align}
where the actual distribution $\mathbb{P}_{a}$ is a reference distribution and can be the public information or estimated according to the historical data. If the EMD value $\small \theta_{h}$ of client $\small h \! \in \!  \{1, 2, \ldots, H\}$ is larger than the pre-set threshold (i.e., $\small \theta_{h} \!\! > \!\! \theta_{\text{th}}$), client $h$ encounters a failure in executing the FL task.

\subsection{Incentive Mechanism Design with Stackelberg Game}
Supposing that $N$ clients are selected by the central server. The interactions between the central server and selected clients are modeled as a two-stage Stackelberg game. Specifically, at Stage \uppercase\expandafter{\romannumeral1}, the central server, which acts as the leader, decides the optimal payment $\small \mathcal{R}_{t}^{*}$ in $t$-th global iteration to minimize its cost $\small \mathcal{U}_{T}$. Then, at Stage \uppercase\expandafter{\romannumeral2}, based on the reward allocated by the central server, each selected client $i \in \{1, 2, \ldots, N\}$, which acts as the follower, maximizes its utility function $\small \mathcal{U}_{i}^{t}$ by determining the optimal privacy budget $\rho_{i}^{t}$. 

\subsubsection{Central Server's Cost (Stage \texorpdfstring {\uppercase\expandafter{\romannumeral1}}{})}

Before introducing the cost function of the central server, we first discuss how the privacy budget $\rho_{i}^{t}$ of each client affects the accuracy loss of the global model. From Eq (\ref{perturb_gradient}), we can derive the global model with Gaussian random noise as follows:
\begin{align}\label{perturb_gradient_noise}
\textstyle \nabla \widetilde{F}(\boldsymbol{w(t)})&= \textstyle  \frac{1}{N} \sum_{i=1}^{N}\left(\nabla F_{i}(\boldsymbol{w(t)}) + \boldsymbol{n_{i}(t)}\right) \nonumber \\ 
&= \textstyle \nabla F(\boldsymbol{w(t)}) + \frac{1}{N}\textstyle\sum_{i=1}^{N} \boldsymbol{n_{i}(t)}. 
\end{align}

Inspired by \cite{rakhlin2011making, yi2022stackelberg}, we assume that the global loss function $\small \nabla F(\boldsymbol{w(t)})$ attains an upper bound $V$ (i.e., $\small \|\nabla F(\boldsymbol{w(t)})\|_{2} \! \leq \! V$). Further, note that the Gaussian random noise possesses zero mean and $\small \mathbb{E}[\|\! \sum_{i=1}^{N} \! \boldsymbol{n_{i}(t)}\|_{2}^{2}] \! = \! d \sigma_{i}^{2}(t)$, where $d$ represents the dimension of the input vector and the variance $\small \sigma_{i}^{2}(t)$ is determined by client $i$'s privacy budget $\small \rho_{i}^{t}$ as shown in Eq (\ref{noise}). Thus, the upper bound of $\small \mathbb{E}[\|\nabla \widetilde{F}(\boldsymbol{w(t)})\|_{2}^{2}]$ can be derived as:
\begin{align}\label{expectation_perturb_gradient}
\textstyle \mathbb{E}[\|\nabla \widetilde{F}(\boldsymbol{w(t)})\|_{2}^{2}] &= \textstyle \mathbb{E}[\|\nabla F(\boldsymbol{w(t)}) + \frac{1}{N} \sum_{i=1}^{N} \boldsymbol{n_{i}(t)} \|_{2}^{2}] \nonumber \\
&= \textstyle \mathbb{E}[\|\nabla F(\boldsymbol{w(t)})\|_{2}^{2}] \!+\! \frac{1}{N^{2}}\mathbb{E}[\|\textstyle \sum_{i=1}^{N}\! \boldsymbol{n_{i}(t)}\|_{2}^{2}]\nonumber \\
&\leq \textstyle V^{2} + \frac{d}{N^{2}}\textstyle \sum_{i=1}^{N} {\sigma_{i}^{2}(t)} \triangleq G^{2}(t). 
\end{align}

Suppose that the global loss function  $\nabla F(\boldsymbol{w(t)})$ satisfies $\beta$-Lipschitz Smoothness (Assumption \ref{smoothness_assumption}) and $\lambda$-Strong Convexity (Assumption \ref{convexity_assumption}), attains minimum at $\small \boldsymbol{w^{*}}$ and $\small \mathbb{E}[\|\nabla \widetilde{F}(\boldsymbol{w(t)})\|_{2}^{2}] \!\! \leq \!\! G^{2}(t)$. Since the upper bound of $\small \mathbb{E}[\|\nabla \widetilde{F}(\boldsymbol{w(t)})\|_{2}^{2}]$ varies with time, we define $\small G^{2} \!=\! \sum_{t=1}^{T}\! G^{2}(t)/T$. Denote the learning rate as $\small \eta_{t}\!=\! 1/\lambda t$, the accuracy loss can be expressed as:
\begin{align}\label{accuracy_loss}
\textstyle \mathbb{E}[F(\boldsymbol{w(T)})-F(\boldsymbol{w^{*}})] \leq \frac{2\beta G^{2}}{\lambda^{2}T}. 
\end{align}

Denote the reward vector by $\small \boldsymbol{\mathcal{R}} \!=\! \{\mathcal{R}_{1}, \mathcal{R}_{2}, \ldots, \mathcal{R}_{T}\}$ and the privacy budget vector as $\small \boldsymbol{\rho} \!=\! \{\rho_{i}^{t}, i \! \in \! \{1,2,\ldots, N\}, t \! \in \! \{1,2,\ldots, \\ T\}\}$. Then, the central server's cost function can be expressed as the summation of the accuracy loss and total reward:
\begin{align}\label{central_server_utility}
& \textstyle \mathcal{U}_{T}(\boldsymbol{\mathcal{R}},\boldsymbol{\rho}) = \gamma \frac{2\beta G^{2}}{\lambda^{2}T} + (1-\gamma) \textstyle \sum_{k=1}^{T} \pi^{k-1} \mathcal{R}_{k}  \\
&= \textstyle \frac{2\beta\gamma}{\lambda^{2}T^{2}} \! \textstyle \sum_{t=1}^{T} (V^{2} \!+\!\! \sum_{i=1}^{N} \! \frac{2 d C^{2}}{|\mathcal{D}_{i}|^{2}\rho_{i}^{t}N^{2}}) \!\!+\! (1 \!-\! \gamma) \! \sum_{k=1}^{T}\! \pi^{k-1} \mathcal{R}_{k},  \nonumber
\end{align}
where discount factor $\small \pi \! \in \! (0,1)$ measures the decrement of the value of reward $\mathcal{R}_{t}$ over time.

\subsubsection{Client's Utility (Stage \texorpdfstring {\uppercase\expandafter{\romannumeral2}}{})}
In the $t$-th global iteration, given the central server's reward $\mathcal{R}_{t}$, the reward $r_{i}^{t}$ allocated to client $i$ is:
$r_{i}^{t} \!=\! \frac{\rho_{i}^{t}}{\sum_{j=1}^{N}\rho_{j}^{t}}\mathcal{R}_{t}$. The utility of client $i$ in $t$-th global iteration is defined as the difference between the reward from the central server and training cost $\mathcal{C}_{i,t}$, i.e.,
\begin{align}
&\mathcal{U}_{i}^{t}(\rho_{i}^{t}, \boldsymbol{\rho_{\text{-}i}^{t}},\mathcal{R}_{t}) = \textstyle \frac{\rho_{i}^{t}}{\sum_{j=1}^{N}\rho_{j}^{t}}\mathcal{R}_{t}-\mathcal{C}_{i,t}, \label{client_utility} \\
&\mathcal{C}_{i,t} = \phi_{1}\mathcal{C}_{i,t}^{pv} + \phi_{2}\mathcal{C}_{i,t}^{d} + \phi_{3}\mathcal{C}_{i,t}^{cp} + \phi_{4}\mathcal{C}_{i,t}^{cm}, \label{client_cost} 
\end{align}
where $\small \boldsymbol{\rho_{\text{-}i}^{t}}\!=\!\{\rho_{1}^{t}, \ldots, \rho_{i-1}^{t}, \rho_{i+1}^{t}, \ldots, \rho_{N}^{t}\}$ and $\small \phi_{k}$, $k \! \in \! \{1, 2, 3, 4\}$ is positive coefficients. The training cost $\small \mathcal{C}_{i,t}$ consists of privacy cost $\small \mathcal{C}_{i,t}^{pv}$, local data cost $\small \mathcal{C}_{i,t}^{d}$, computation cost $\small \mathcal{C}_{i,t}^{cp}$ and communication cost $\small \mathcal{C}_{i,t}^{cm}$. Among all components of the training cost of the client $\small i$, privacy cost $\small \mathcal{C}_{i,t}^{pv}$ is closely related to the privacy budget $\small \rho_{i}^{t}$, as a larger privacy budget $\small \rho_{i}^{t}$ signifies more precise data, yet it also corresponds to an increased vulnerability to privacy breaches.  Specifically, inspired by \cite{hu2020trading}, we denote the privacy cost of client $i$ as a function of $\small c(\nu_{i}^{t}, \rho_{i}^{t})$, where $\small \nu_{i}^{t} \! > \! 0$ is the privacy value that is assumed publicly known. Here, we consider linear privacy cost function for each client, i.e., $\mathcal{C}_{i,t}^{pv} = c(\nu_{i}^{t}, \rho_{i}^{t}) = \nu_{i}^{t} \cdot \rho_{i}^{t}$. 


The objective of $i$-th client is to maximize its utility function $\mathcal{U}_{i}^{t}$ by dynamically adjusting privacy budgets $\rho_{i}^{t}$ at each global iteration $t$ according to the reward $\mathcal{R}_{t}$ from the central server. The goal of the central server is to minimize its cost function by adjusting the reward $\mathcal{R}_{t}$, $t \! \in \! \{1, 2, \ldots, T\}$ distributed to participating clients in the premise of reaching preset global model accuracy. Note that the reward $\boldsymbol{\mathcal{R}}$ to clients will affect the clients' designs of privacy budget $\boldsymbol{\rho}$, which in return affects the central server's cost as shown in Eq (\ref{central_server_utility}). Thus, the two-stage Stackelberg game can be formulated as:
\begin{align}
\text{Stage \uppercase\expandafter{\romannumeral1}:} &\textstyle \ \min_{\boldsymbol{\mathcal{R}}} \mathcal{U}_{T}(\boldsymbol{\mathcal{R}},\boldsymbol{\rho}), \label{central_server_obj} \\
\text{Stage \uppercase\expandafter{\romannumeral2}:} &\textstyle \ \max_{\rho_{i}^{t}}\mathcal{U}_{i}^{t}(\rho_{i}^{t}, \boldsymbol{\rho_{\text{-}i}^{t}},\mathcal{R}_{t}), \label{client_obj}  
\end{align}
where the privacy budget $\small \rho_{i}^{t}$ over time form the strategy profile $\small \boldsymbol{\rho_{i}}$ of client $i$, i.e.,  $\small \boldsymbol{\rho_{i}} =\{\rho_{i}^{t},t \in \{1,2, \ldots, T\}\}$. 

\section{Stackelberg Nash  Equilibrium Analysis}
In this section, we will find the optimal strategy profile $\small (\boldsymbol{\mathcal{R}}, \\ \{\boldsymbol{\rho_{i}}$, $i \! \in \! \{1,2, \ldots, N\}\})$, with $\small \boldsymbol{\mathcal{R}} \!=\! \{\mathcal{R}_{t}, t \! \in \! \{1,2, \ldots,  T\}\}\!$ and $\small \boldsymbol{\rho_{i}} \!=\! \{\rho_{i}^{t},t \! \in \! \{1,2, \ldots, T\}\}$ for the central server and selected clients in the two-stage Stackelberg game through backward induction. Firstly, we concentrate on the followers' operation and derive each selected client $i$’s optimal privacy budget $\rho_{i}^{t*}$ in the $t$-th global iteration under any given reward $\mathcal{R}_{t}$. Then, considering the trade-off between the model accuracy loss and payment to the clients, we deduce the central server’s optimal strategy $\mathcal{R}_{t}^{*}$. Finally, we prove that the optimal solution forms the Stackelberg Nash Equilibrium. 

\subsection{Optimal Strategy Profile}
We adopt a backward induction approach to derive the optimal strategy of the central server and each client respectively. First of all, we analyze the optimal strategy of each selected client by determining the optimal privacy budget $\rho_{i}^{t*}$ and are presented in the following theorem.

\begin{theorem}\label{theorem_optimal_rho}
\emph{In Stage \uppercase\expandafter{\romannumeral2}, given payment $\mathcal{R}_{t}$ in the $t$-th global iteration, the optimal privacy budget of each client $i$ is:}
\begin{align}\label{optimal_rho}
\textstyle \rho_{i}^{t*}=\frac{\mathcal{R}_{t}(N-1)}{\phi_{1}\sum_{i=1}^{N}\nu_{i}^{t}} - \frac{\mathcal{R}_{t}\nu_{i}^{t}}{\phi_{1}}(\frac{N-1}{\sum_{i=1}^{N}\nu_{i}^{t}})^{2}. 
\end{align}
\end{theorem}
\begin{proof}
Firstly, we derive the first-order and the second-order derivatives of each client $i$'s utility function $\mathcal{U}_{i}^{t}(\rho_{i}^{t}, \boldsymbol{\rho_{\text{-}i}^{t}},\mathcal{R}_{t})$ concerning privacy budget $\rho_{i}^{t}$ as follows:
\begin{align}
&\textstyle \frac{\partial \mathcal{U}_{i}^{t}}{\partial \rho_{i}^{t}}\!=\! \frac{\mathcal{R}_{t}\! \sum_{j=1}^{N}\!\rho_{j}^{t} \!-\! \rho_{i}^{t}}{(\sum_{i=1}^{N}\!\rho_{i}^{t})^{2}} \!-\! \phi_{1}\nu_{i}^{t} \!=\! \frac{\mathcal{R}_{t}\! \sum_{j\in N \setminus \{i\}}\!\rho_{j}^{t}}{(\sum_{i=1}^{N}\!\rho_{i}^{t})^{2}} \!-\! \phi_{1}\nu_{i}^{t}, \label{client_utility_first_order} \\
&\textstyle \frac{\partial^{2} \mathcal{U}_{i}^{t}}{\partial (\rho_{i}^{t})^{2}}= - 2\mathcal{R}_{t} \cdot \frac{\sum_{j\in N \setminus \{i\}}\rho_{j}^{t}}{(\sum_{i=1}^{N}\rho_{i}^{t})^{3}} < 0. \label{client_utility_second_order}
\end{align}

As the second-order derivative $\small \partial^{2} \mathcal{U}_{i}^{t}/\partial (\rho_{i}^{t})^{2} \!>\! 0$, the utility of each client $i$ is strictly concave in the feasible region of $\rho_{i}^{t}$. Then, we derive the optimal privacy budget of client $i$ in $t$-th global iteration by solving equation $\partial \mathcal{U}_{i}^{t}/\partial \rho_{i}^{t}=0$, i.e.,
\begin{align}\label{optimal_rho_with_other_rho}
\textstyle \rho_{i}^{t*}=\sum_{i=1}^{N}\rho_{i}^{t} - \frac{\phi_{1}\nu_{i}^{t}}{\mathcal{R}_{t}}(\sum_{i=1}^{N}\rho_{i}^{t})^{2}.
\end{align}

From Eq (\ref{client_utility_first_order}), since client $i$ fails to acquire other clients' privacy budget $\small \rho_{j}^{t}, j \!\in\! N \!\setminus\! \{i\}$ as it is privacy information and there is no communication among clients during any consecutive global training iterations, we need to derive $\sum_{i=1}^{N}\rho_{i}^{t}$ term. Thus, Eq (\ref{client_utility_first_order}) can be rewritten as follows:
\begin{align}\label{sum_rho_except_i}
\textstyle \sum_{j\in N \setminus \{i\}}\rho_{j}^{t} =\frac{\phi_{1}\nu_{i}^{t}}{\mathcal{R}_{t}}(\sum_{i=1}^{N}\rho_{i}^{t})^{2}.
\end{align}

After summing up both sides of Eq (\ref{sum_rho_except_i}), we have:
\begin{align}
\textstyle \sum_{i=1}^{N} [\textstyle \sum_{j\in N \setminus \{i\}}\rho_{j}^{t}] &= \textstyle \sum_{i=1}^{N} \frac{\phi_{1}\nu_{i}^{t}}{\mathcal{R}_{t}} (\sum_{i=1}^{N}\rho_{i}^{t})^{2}, \nonumber \\
\Rightarrow \ \textstyle (N\!-\!1)\sum_{i=1}^{N}\rho_{i}^{t} &= \textstyle \frac{\phi_{1}}{\mathcal{R}_{t}}(\sum_{i=1}^{N}\rho_{i}^{t})^{2}\sum_{i=1}^{N}\nu_{i}^{t}, \nonumber \\
\Rightarrow \ \textstyle \sum_{i=1}^{N}\rho_{i}^{t} &= \textstyle \frac{(N-1)\mathcal{R}_{t}}{\phi_{1}\sum_{i=1}^{N}\nu_{i}^{t}}. 
\end{align}

By substituting the term $\sum_{i=1}^{N}\rho_{i}^{t}$ in Eq (\ref{optimal_rho_with_other_rho}), we can finalize the expression of the optimal privacy budget $\rho_{i}^{t*}$ as:
\begin{align}
\rho_{i}^{t*} = \textstyle \frac{\mathcal{R}_{t}(N-1)}{\phi_{1}\sum_{i=1}^{N}\nu_{i}^{t}} \!-\! \frac{\mathcal{R}_{t}\nu_{i}^{t}}{\phi_{1}}(\frac{N-1}{\sum_{i=1}^{N}\nu_{i}^{t}})^{2}. \nonumber
\end{align} 

Hence, the theorem holds. \end{proof}

Based on the optimal privacy budget $\rho_{i}^{t*}$, we derive the central server's optimal reward $\mathcal{R}_{t}^{*}$ as summarized in Theorem~\ref{theorem_optimal_reward}. 
\begin{theorem}\label{theorem_optimal_reward}
\emph{In Stage \uppercase\expandafter{\romannumeral1}, based on the optimal privacy budget of each client $i$, i.e, $\small \boldsymbol{\rho_{i}^{t*}} \!=\! \{\rho_{1}^{t*}, \rho_{2}^{t*}, \ldots, \rho_{N}^{t*}\}$, the optimal reward of each global training iteration $\small \mathcal{R}_{t}^{*} \!=\! \sqrt{\frac{A_{t}\pi^{1-t}}{1-\gamma}}$, where discount factor $\small \pi \! \in \! (0,1)$ measures the decrement of the value of reward over time, constants $\small C_{i}\!=\!\frac{(N-1)\left[(\sum_{i=1}^{N}\nu_{i}^{t})-\nu_{i}^{t}(N-1)\right]}{\phi_{1}\left(\sum_{i=1}^{N}\nu_{i}^{t}\right)^{2}}$ and $\small A_{t}\!=\!\sum_{i=1}^{N} \frac{4d\gamma\beta C^{2}}{T^{2}\lambda^{2} N^{2}C_{i}|\mathcal{D}_{i}|^{2}}$.}
\end{theorem}
\begin{proof}
First of all, by substituting the optimal privacy budget $\rho_{i}^{t*}$ in Eq~(\ref{optimal_rho}) into the central server's cost function in Eq~(\ref{central_server_utility}), the cost function $\small \mathcal{U}_{T}(\boldsymbol{\mathcal{R}}, \boldsymbol{\rho}^{*})$ can be rewritten as follows:
\begin{align}\label{central_server_utility_with_optimal_rho}
\mathcal{U}_{T}(\boldsymbol{\mathcal{R}}, \boldsymbol{\rho^{*}})
&= \textstyle \frac{2\beta\gamma}{\lambda^{2}T^{2}} \sum_{t=1}^{T} (V^{2} \!+\! \frac{d}{N^{2}}\sum_{i=1}^{N} \frac{2 C^{2}}{|\mathcal{D}_{i}|^{2}\rho_{i}^{t*}}) \nonumber \\ 
& \quad \textstyle + (1-\gamma) \sum_{k=1}^{T}\pi^{k-1} \mathcal{R}_{k}, \nonumber \\ 
&= \textstyle \frac{2\beta\gamma}{\lambda^{2}T^{2}} (V^{2}T + \frac{2d C^{2}}{N^{2}} \textstyle \sum_{t=1}^{T} \sum_{i=1}^{N} \frac{1}{|\mathcal{D}_{i}|^{2}\rho_{i}^{t*}}) \nonumber \\ 
& \textstyle \quad + (1-\gamma) \textstyle \sum_{k=1}^{T}\pi^{k-1} \mathcal{R}_{k}, 
\end{align}
where $\small \boldsymbol{\rho^{*}} \!=\! \{\rho_{i}^{t*}, i \! \in \! \{1,2,\ldots, N\}, t \! \in \! \{1,2,\ldots, T\}\}$. The first-order of the central server's cost function can be derived as:
\begin{align}\label{central_server_utility_first_order}
\textstyle \frac{\partial \mathcal{U}_{T}}{\partial \mathcal{R}_{t}} = (1 \!-\! \gamma) \pi^{t-1} \!-\! \frac{4 d \gamma \beta C^{2}}{(T \lambda N)^{2}}  \sum_{i=1}^{N} \frac{\partial \rho_{i}^{t*}/\partial \mathcal{R}_{t}}{(|\mathcal{D}_{i}| \rho_{i}^{t*})^{2}}.
\end{align}

Then, we need to consider the existence of the solution of Eq (\ref{central_server_utility_first_order}). As the reward $\mathcal{R}_{t} \in [0, + \infty)$, we have:
\begin{align}\label{limit}
\textstyle \lim_{\mathcal{R}_{t} \rightarrow 0} \frac{\partial \mathcal{U}_{T}}{\partial \mathcal{R}_{t}} \!\rightarrow\! - \infty, \lim_{\mathcal{R}_{t} \rightarrow + \infty} \frac{\partial \mathcal{U}_{T}}{\partial \mathcal{R}_{t}} \!\rightarrow\! (1\!-\!\gamma)\pi^{t} \!>\! 0. 
\end{align}

From Eq (\ref{limit}), we demonstrate the existence of a solution for the equation $\partial \mathcal{U}_{T}/\partial \mathcal{R}_{t}\!=\!0$. Further, according to Eq (\ref{central_server_utility_first_order}), we derive the second-order of $\small \mathcal{U}_{T}(\boldsymbol{\mathcal{R}},\boldsymbol{\rho^{*}})$ with respect to $\mathcal{R}_{t}$:
\begin{align}\label{central_server_utility_second_order}
\textstyle \frac{\partial^{2} \mathcal{U}_{T}}{\partial \mathcal{R}_{t}^{2}} = \frac{8 d \gamma \beta C^{2}}{(T\lambda N)^{2}} \sum_{i=1}^{N} (\frac{\partial \rho_{i}^{t*}}{\partial \mathcal{R}_{t}}\frac{1}{|\mathcal{D}_{i}|})^{2}\! (\frac{1}{\rho_{i}^{t*}})^{3} > 0. 
\end{align}

Thus, based on Eqs (\ref{limit})-(\ref{central_server_utility_second_order}), we obtain the optimal reward $\mathcal{R}_{t}^{*}$ by solving the first-order equation $\small \partial \mathcal{U}_{T}/\partial \mathcal{R}_{t}\!=\!0$, i.e.,
\begin{align}
\textstyle \frac{\partial \mathcal{U}_{T}}{\partial \mathcal{R}_{t}} &= \textstyle (1 \!-\! \gamma) \pi^{t-1} \!-\! \frac{4 d \gamma \beta C^{2}}{(T \lambda N)^{2}} \sum_{i=1}^{N} \frac{\partial \rho_{i}^{t*}/\partial \mathcal{R}_{t}}{(|\mathcal{D}_{i}| \rho_{i}^{t*})^{2}} \nonumber \\
&= \textstyle (1 \!-\! \gamma)\pi^{t-1} \!-\! \frac{A_{t}}{\mathcal{R}_{t}^{2}}=0 
\Rightarrow \! \mathcal{R}_{t} \!=\! \sqrt{\frac{A_{t}}{(1-\gamma)\pi^{t-1}}}, \nonumber
\end{align}
where known constants $C_{i}\!=\!\frac{(N-1)\left[(\sum_{i=1}^{N}\nu_{i}^{t})-\nu_{i}^{t}(N-1)\right]}{\phi_{1}\left(\sum_{i=1}^{N}\nu_{i}^{t}\right)^{2}}$ and $A_{t} \\ =\sum_{i=1}^{N} \frac{4d\gamma\beta C^{2}}{T^{2}\lambda^{2} N^{2}C_{i}|\mathcal{D}_{i}|^{2}}$.  \end{proof}

As shown in Theorem \ref{theorem_optimal_reward}, the optimal reward $\mathcal{R}_{t}^{*}$ increases with $t$, indicating that the central server necessitates data of superior quality to attain the desired accuracy level when approaching the end of the time horizon $T$.

\begin{definition}\label{definition2}
\emph{(Stackelberg Nash Equilibrium) The strategy profile $(\boldsymbol{\mathcal{R}}^{*},\{\boldsymbol{\rho_{i}^{*}}$, $i \!\in\! \{1,2, \ldots, N\}\})$, with $\boldsymbol{\mathcal{R}}^{*} \!=\! \{\mathcal{R}_{t}^{*}, t \in \{1, \\ 2,\ldots,T\}\}$ and $\boldsymbol{\rho_{i}^{*}} \!\!=\!\! \{\rho_{i}^{t*},t \in \{1,2,\ldots, T\} \}$ constitutes a Stackelberg Nash Equilibrium if for any reward $\boldsymbol{\mathcal{R}} \! \in \! \mathbb{R}^{T}$ and any privacy budget $\rho_{i}^{t} \geq 0$:}
\begin{align}\label{eq29&30}
\mathcal{U}_{T}(\boldsymbol{\mathcal{R}^{*}}, \boldsymbol{\rho^{*}}) &\leq \mathcal{U}_{T}(\boldsymbol{\mathcal{R}}, \boldsymbol{\rho^{*}}),  \\
\mathcal{U}_{i}^{t}(\rho_{i}^{t*}, \boldsymbol{\rho_{\text{-}i}^{t*}}, \mathcal{R}_{t}^{*}) &\geq \mathcal{U}_{i}^{t}(\rho_{i}^{t}, \boldsymbol{\rho_{\text{-}i}^{t*}}, \mathcal{R}_{t}^{*}). 
\end{align}
\end{definition}

\begin{figure}[t]
\setlength{\abovecaptionskip}{6pt} 
    \centering
    \subfloat[Fitting Curve on MNIST]{
        \label{fit_mnist}
        \includegraphics[width=0.24\textwidth, trim=50 30 60 50,clip]{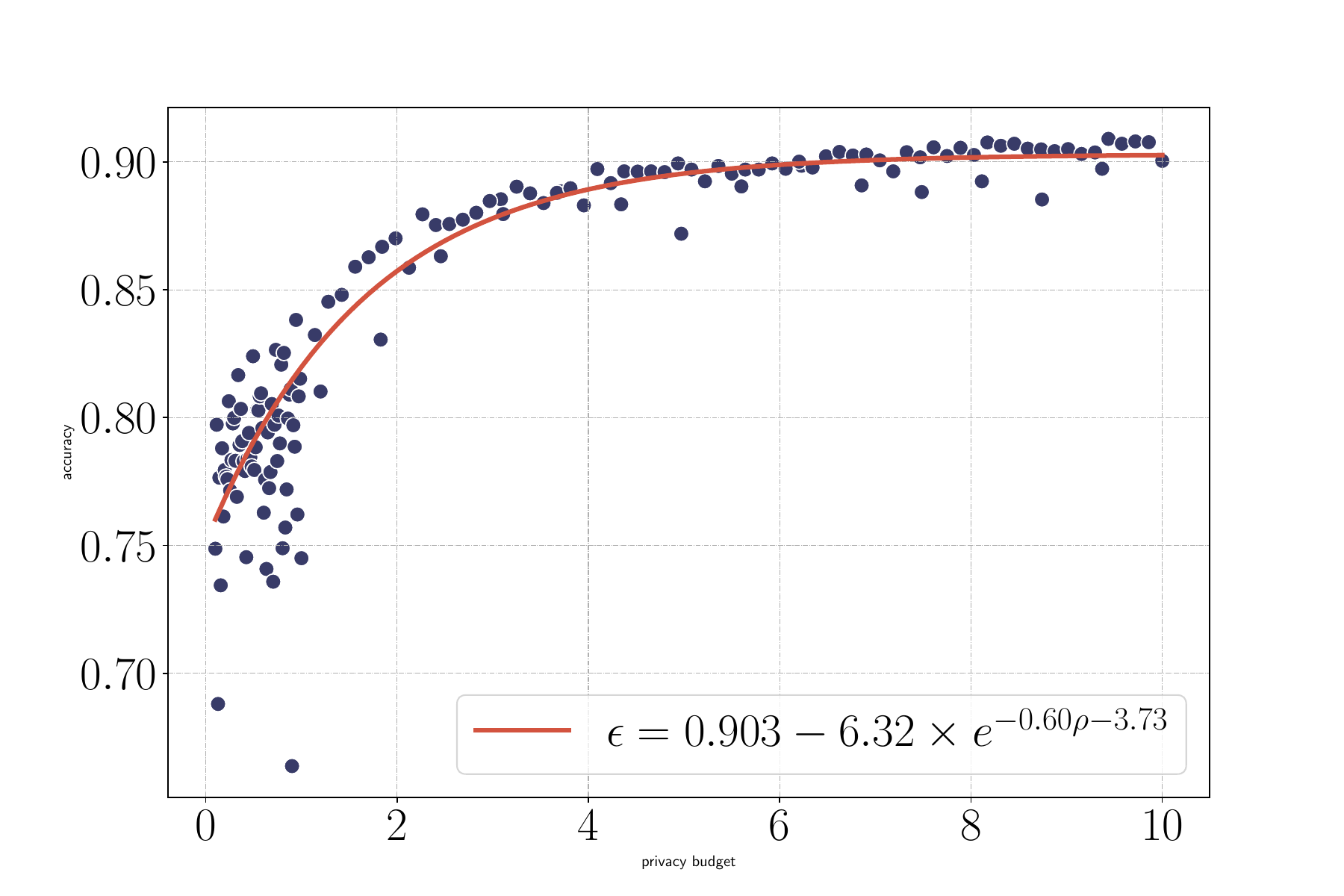}}
    \subfloat[Fitting Curve on EMNIST]{
        \label{fit_emnist}
        \includegraphics[width=0.24\textwidth, trim=60 40 60 50,clip]{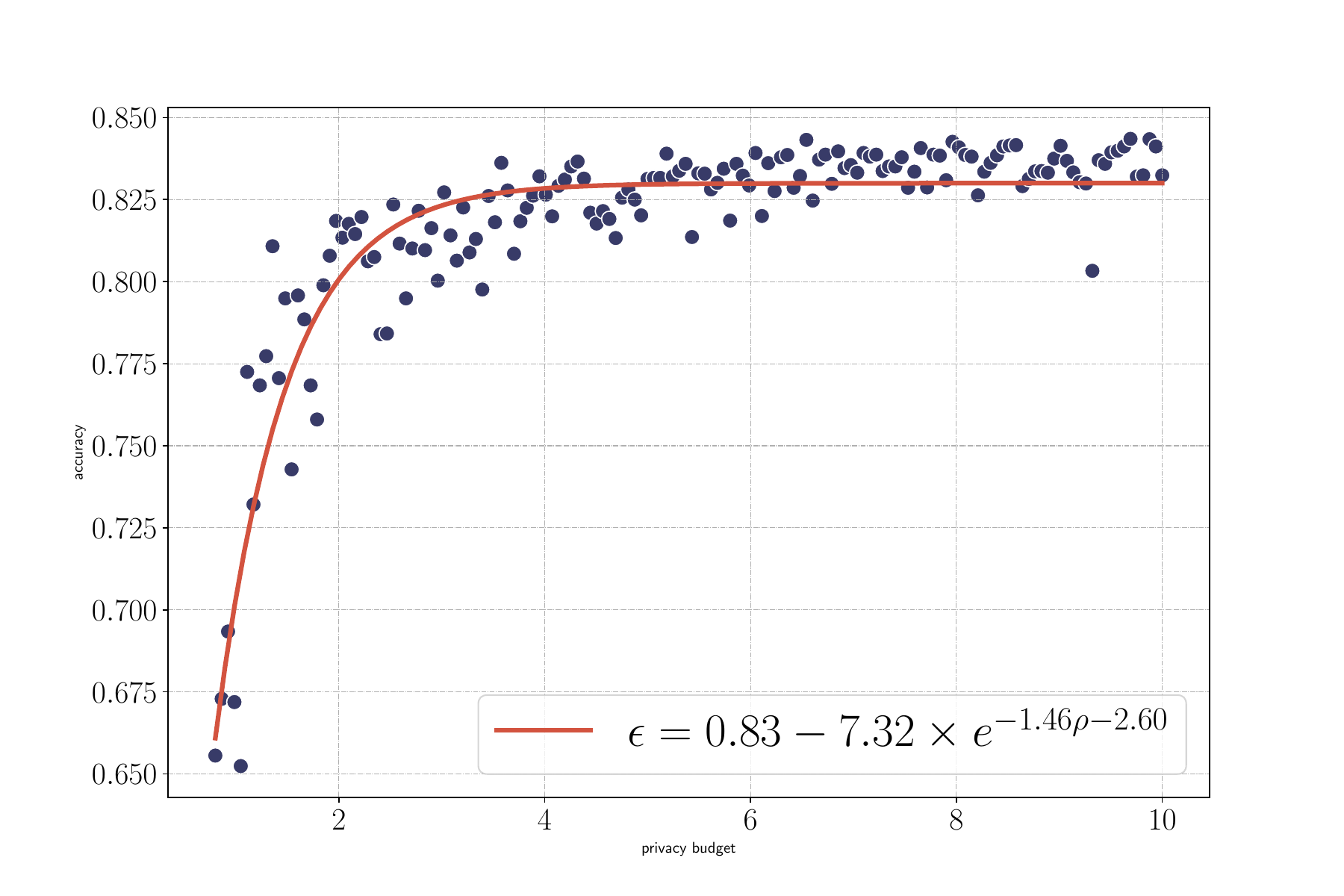}}
\caption{The relationship between the privacy budget $\rho$ and model accuracy on MNIST and EMNIST dataset}
\label{fit}
\vspace{-13pt}    
\end{figure}

\begin{theorem}\label{thm3}
\emph{The above two-stage Stackelberg game possesses a Stackelberg Nash Equilibrium.}
\end{theorem}
\begin{proof}
Based on Theorem \ref{theorem_optimal_rho}, we deduce that there exists the optimal strategy of privacy budget $\rho_{i}^{t*}$ for each client $i \! \in \! \{1, 2, \\ \ldots, N\}$ given any reward $\mathcal{R}_{t} \!\in\! \mathbb{R}$. Then, we need to prove that there exists an optimal reward $\mathcal{R}_{t}^{*}$ for the central server's cost function $\small \mathcal{U}_{T}$ in the premise of optimal privacy budget vector $\small \boldsymbol{\rho_{i}^{t*}}\!=\!\{\rho_{i}^{t*}, i \! \in \! \{1, 2, \ldots, N\}\}$. According to the proof of Theorem \ref{theorem_optimal_reward}, $\small \partial^{2} \mathcal{U}_{T}/\partial \mathcal{R}_{t}^{2} \! > \! 0$ indicates that $\small  \mathcal{U}_{T}(\boldsymbol{\mathcal{R}}, \boldsymbol{\rho^{*}})$ is convex. Thus, an optimal reward $\small \mathcal{R}_{t}^{*}$ for central server exists by solving the first-order equation $\small \partial \mathcal{U}_{T}/\partial \mathcal{R}_{t} \! = \! 0$ based on the optimal privacy budget vector $\small \boldsymbol{\rho_{i}^{t*}}\!=\!\{\rho_{i}^{t*}, i \! \in \! \{1,2, \ldots, N\}\}$. In other words, the strategy $\mathcal{R}_{t}^{*}$ and $\boldsymbol{\rho_{i}^{t*}}$ are mutual optimal strategy for each selected client $i$ and the central server. Thus, the optimal strategy profile $\small (\boldsymbol{\mathcal{R}^{*}}, \{\boldsymbol{\rho_{i}^{*}}$, $i \! \in \! \{1,2, \ldots, N\}\})$, with $\small \boldsymbol{\mathcal{R}^{*}} \!\!=\! \{\mathcal{R}_{t}^{*}, t \! \in \! \{1, 2, \ldots, T \}\}\!$ and $\small \boldsymbol{\rho_{i}^{*}} \!=\! \{\rho_{i}^{t*},t \! \in \! \{1,2, \ldots, T\}\}$ of the two-stage Stackelberg game possesses a Stackelberg Nash Equilibrium. Hence, the theorem holds. \end{proof}

Theorem \ref{theorem_optimal_reward} reveals that the optimal privacy budget of the selected clients in Theorem \ref{theorem_optimal_rho} and optimal reward of the central server in Theorem \ref{theorem_optimal_reward} are mutually optimal, which leads to the steady state of the FL system. The overall framework process is summarized in Algorithm \ref{alg1}.

\begin{algorithm}[t]
\small
\caption{QI-DPFL} 
\label{alg1}
\hspace*{0.02in} {\bf Input: } 
The total client size $H$, hyperparameter $\gamma$ and $\phi_{k}$, $k \!\in\! \{1, 2, 3, 4\}$, discount factor $\pi \!\in\! (0,1)$, local training epoch $L$, global training iteration $T$ and convex function $I$.
\begin{algorithmic}[1]
\State The central server claims the task request by outsourcing the FL task and initializes the global model.
\State \textbf{Client Selection Layer:}
\For{each client $\small h \in \{1, 2, \ldots, H\}$}
\State Calculate Wasserstein distance
\If{$\small \theta_{h} < \theta_{\text{th}}$}
\State Add $h$-th client into the selected client set
\EndIf
\EndFor
\State \textbf{Stackelberg Game-based Model Training Layer:}
\For{$t=0$ to $T$}
    \For{selected client $\small i \!\in\!\{1,2, \ldots, N\}$ in parallel}
    \State Collect privacy data and perform local training:
        \For {$l=1$ to $L$}
        \State $\small \boldsymbol{w_{i}^{l+1}(t)}=\boldsymbol{w_{i}^{l}(t)}-\eta_{i} \nabla F_{i}(\boldsymbol{w_{i}^{l}(t)})$
        \EndFor
    \State Determine the optimal privacy budget $\small \rho_{i}^{t*}$
    \State Perturb the local model and upload $\small \nabla \widetilde{F}_{i}(\boldsymbol{w_{i}(t)})$
    \EndFor
\State $\small \textstyle \boldsymbol{w(t+1)}=\boldsymbol{w(t)}-\eta \cdot \frac{1}{N} \sum_{i=1}^{N}  \nabla \widetilde{F}_{i}(\boldsymbol{w(t)})$
\State Distribute reward $\small \mathcal{R}_{t}^{*}$ to all selected clients
\EndFor 
\end{algorithmic}
\end{algorithm}

\subsection{Range Analysis of the Reward}

In this section, we will analyze the range of the reward under the Stackelberg Nash Equilibrium to guarantee the pre-set model accuracy while simultaneously minimizing the compensation provided to clients. Firstly, we introduce the additivity property of the privacy budget $\rho_{i}^{t}$  of each selected client $i$ as summarized in Lemma \ref{lemma1} \cite{hu2020trading}.

\begin{lemma}\label{lemma1}
\emph{Hypothesise two mechanism satisfy $\rho_{1}$-$z$CDP and $\rho_{2}$-$z$CDP, their composition satisfies $\rho_{1}+\rho_{2}$-$z$CDP.}
\end{lemma}

According to Lemma 1, in the $t$-th global iteration, it is equivalent that the global model satisfies $\rho_{t}$-$z$CDP, where the global privacy budget $\rho_{t}$ equals the summation of each client's privacy budget $\small \rho_{i}^{t}$ (i.e., $\small \rho_{t} \!=\! \sum_{i=1}^{N}\! \rho_{i}^{t}$). Naturally, the larger privacy budget usually results in better model performance. To quantify the global model accuracy as a function of the privacy budget, we acquire the test accuracy of the training model with different privacy budgets based on different real-world datasets (i.e., MNIST, Cifar10, and EMNIST datasets). Fig. \ref{fit} illustrates the MNIST and EMNIST datasets as an example. We observe that, in $t$-th global iteration, the global model accuracy $\epsilon_{t}$ can be regarded as a simplified concave function concerning the global privacy budget $\rho_{t}$ (i.e., $\epsilon_{t} \!=\! I(\rho_{t}) \!=\! I_{1} \!-\! I_{2} \!\times\! e^{-I_{3}\rho_{t} - I_{4}}$), where $\small I_{k}$, $\small k \!\in\! \{1,2,3, 4\}$ are corresponding constants.

\begin{remark}
\emph{(Reward Range) Based on the optimal strategy of each selected client, the global privacy budget $\rho_{t}$ can be calculated as:
$\rho_{t} \!=\! \frac{\mathcal{R}_{t}(N-1)}{\phi_{1}\sum_{i=1}^{N}\nu_{i}^{t}}$, based on which, the model accuracy $\epsilon_{t}$ at $t$-th global iteration can be derived as follows: }
\begin{align}\label{model_accuracy}
\epsilon_{t} &= I_{1} - I_{2} \times e^{-I_{3}\rho_{t} - I_{4}} \nonumber \\[-6pt]
&= I_{1} - I_{2} \times e^{-\frac{I_{3}\mathcal{R}_{t}(N-1)}{\phi_{1}\sum_{i=1}^{N}\nu_{i}^{t}} - I_{4}} \geq \epsilon, 
\end{align}
\emph{where $\epsilon$ is the predefined model accuracy that should be reached after $T$ iterations of model training. Thus, there exists a lower and an upper bound for the allocated reward $\mathcal{R}_{t}$ in the $t$-th global iteration, i.e.,}
\begin{align}\label{reward_range}
\textstyle \frac{\phi_{1}[\log (\frac{I_{2}}{I_{1}-\epsilon})-I_{4}]}{(N-1)I_{3}/\sum_{i=1}^{N}\!\nu_{i}^{t}} \leq \mathcal{R}_{t} \leq  \frac{\phi_{1}[\log (\frac{I_{2}}{I_{1}-\epsilon_{\text{max}}})-I_{4}]}{(N-1)I_{3}/\sum_{i=1}^{N}\!\nu_{i}^{t}}.
\end{align}
\end{remark}

\begin{figure*}
\setlength{\abovecaptionskip}{2pt} 
    \centering
    \subfloat[Accuracy on different models]{
        \label{EMNIST_IID_acc_different_model}
        \includegraphics[width=0.249\textwidth, trim=10 0 75 20,clip]{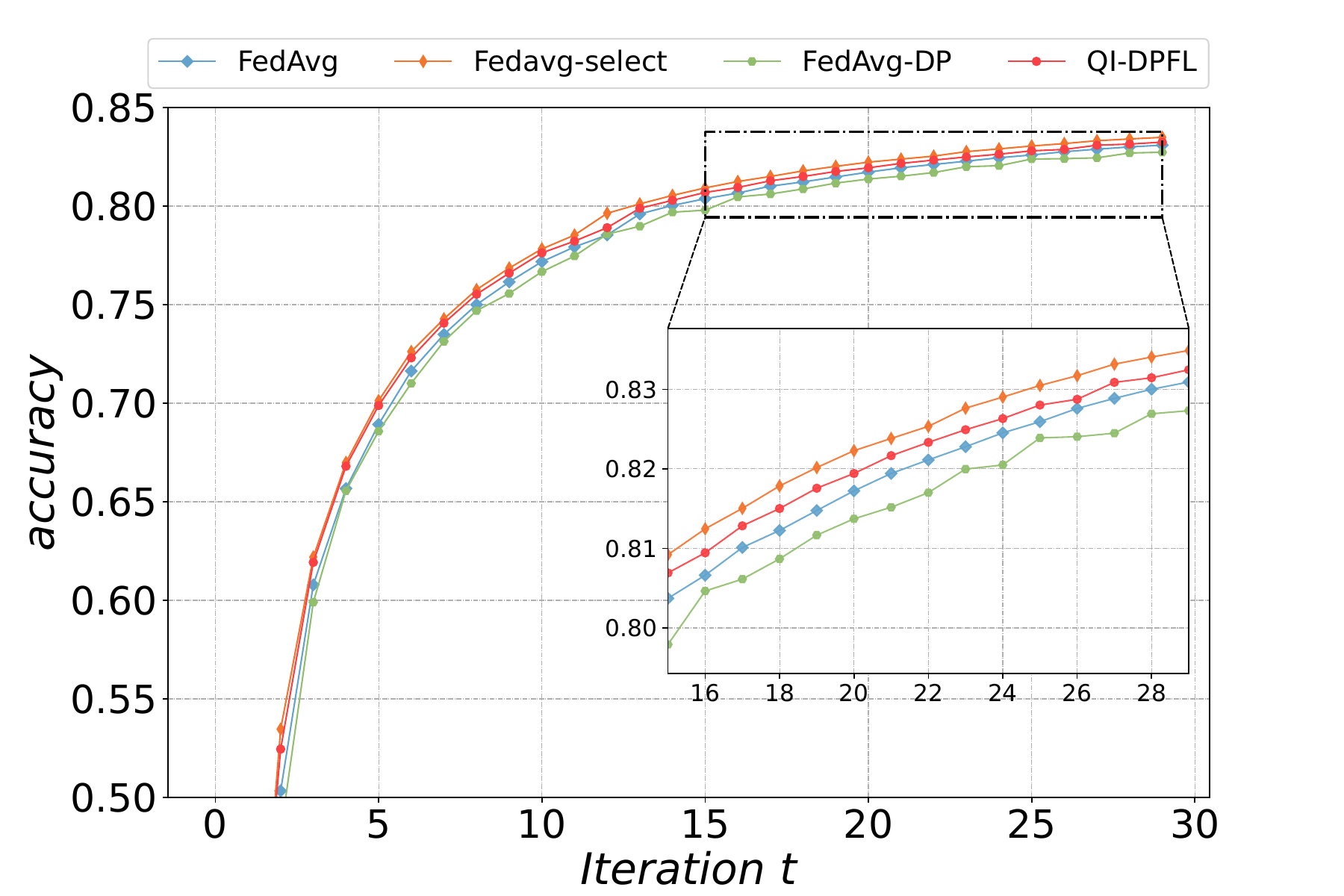}}
    \subfloat[Accuracy on different strategies]{
        \label{EMNIST_IID_acc_different_strategy}
        \includegraphics[width=0.249\textwidth, trim=10 0 75 20,clip]{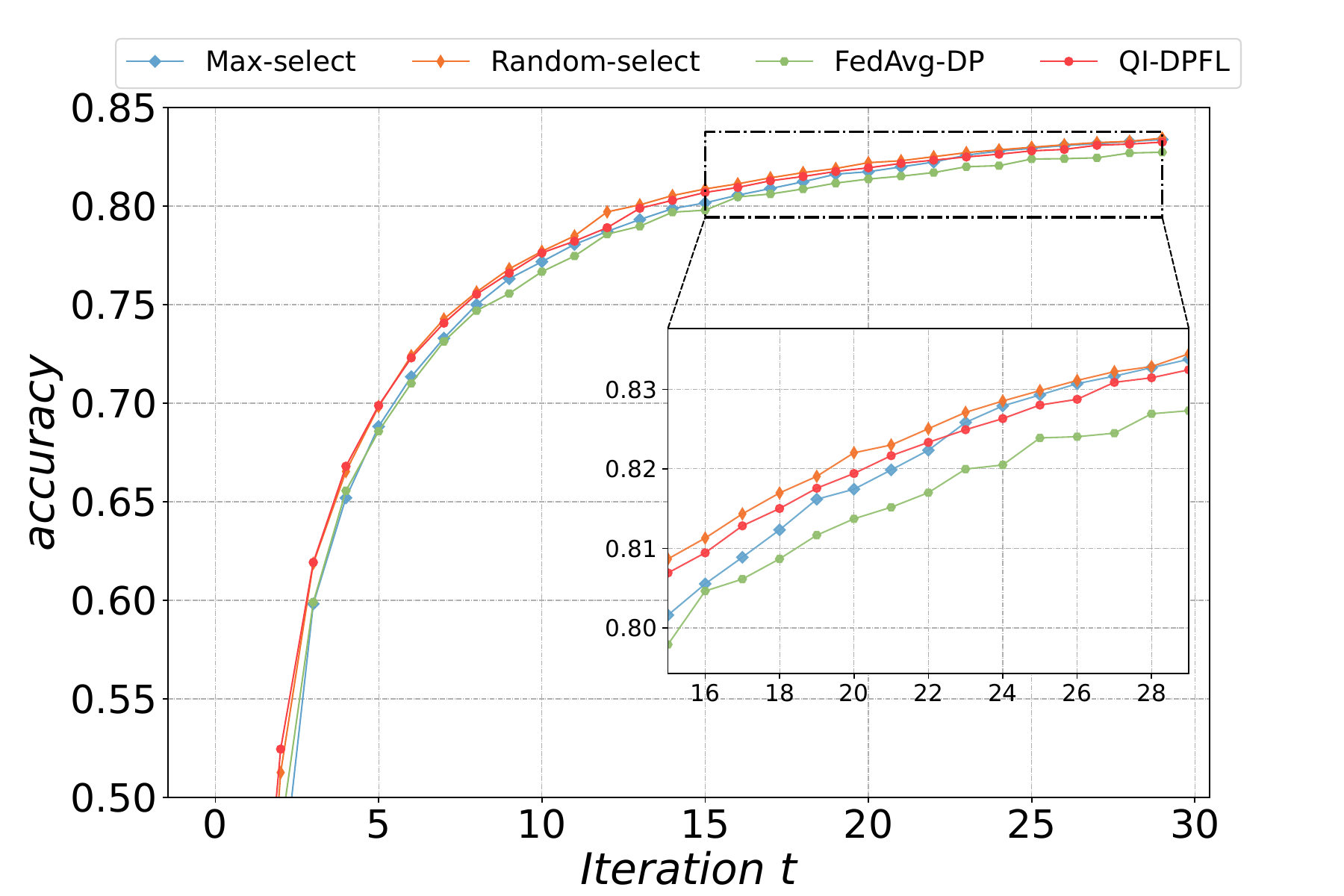}}
    \subfloat[Cost of central server]{
        \label{EMNIST_IID_cost}
        \includegraphics[width=0.249\textwidth, trim=10 0 75 15,clip]{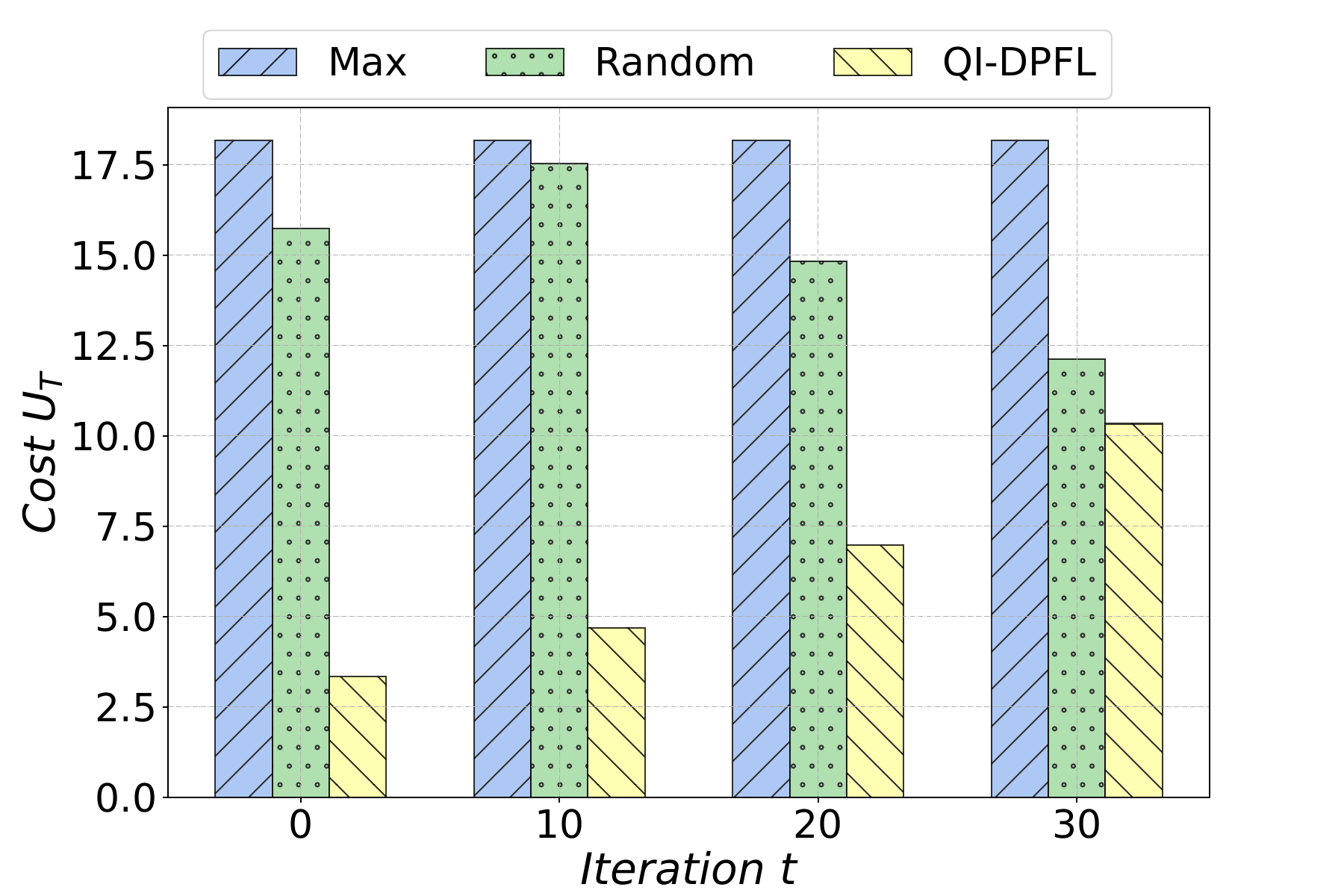}}
    \subfloat[Accuracy on different strategies]{
        \label{EMNIST_IID_reward}
        \includegraphics[width=0.249\textwidth, trim=10 0 75 15,clip]{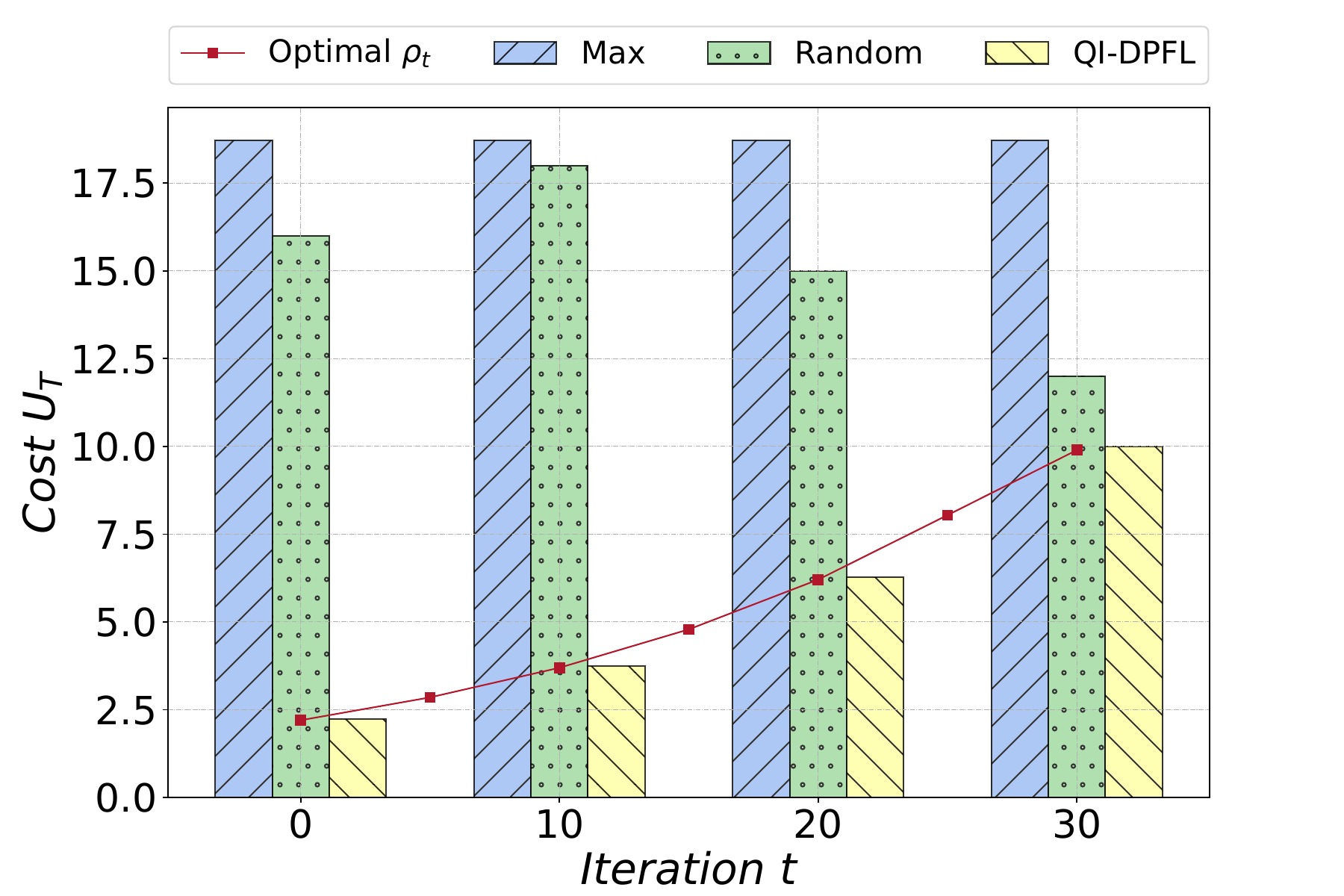}}
    \caption{The performance of different models on EMNIST dataset with IID data distribution.}
    \label{EMNIST_IID}
\vspace{-18pt}    
\end{figure*}

\begin{figure*}
\setlength{\abovecaptionskip}{2pt} 
    \centering
    \subfloat[Accuracy on different models]{
        \label{EMNIST_Non_IID_acc_different_model}
        \includegraphics[width=0.249\textwidth, trim=35 0 75 20,clip]{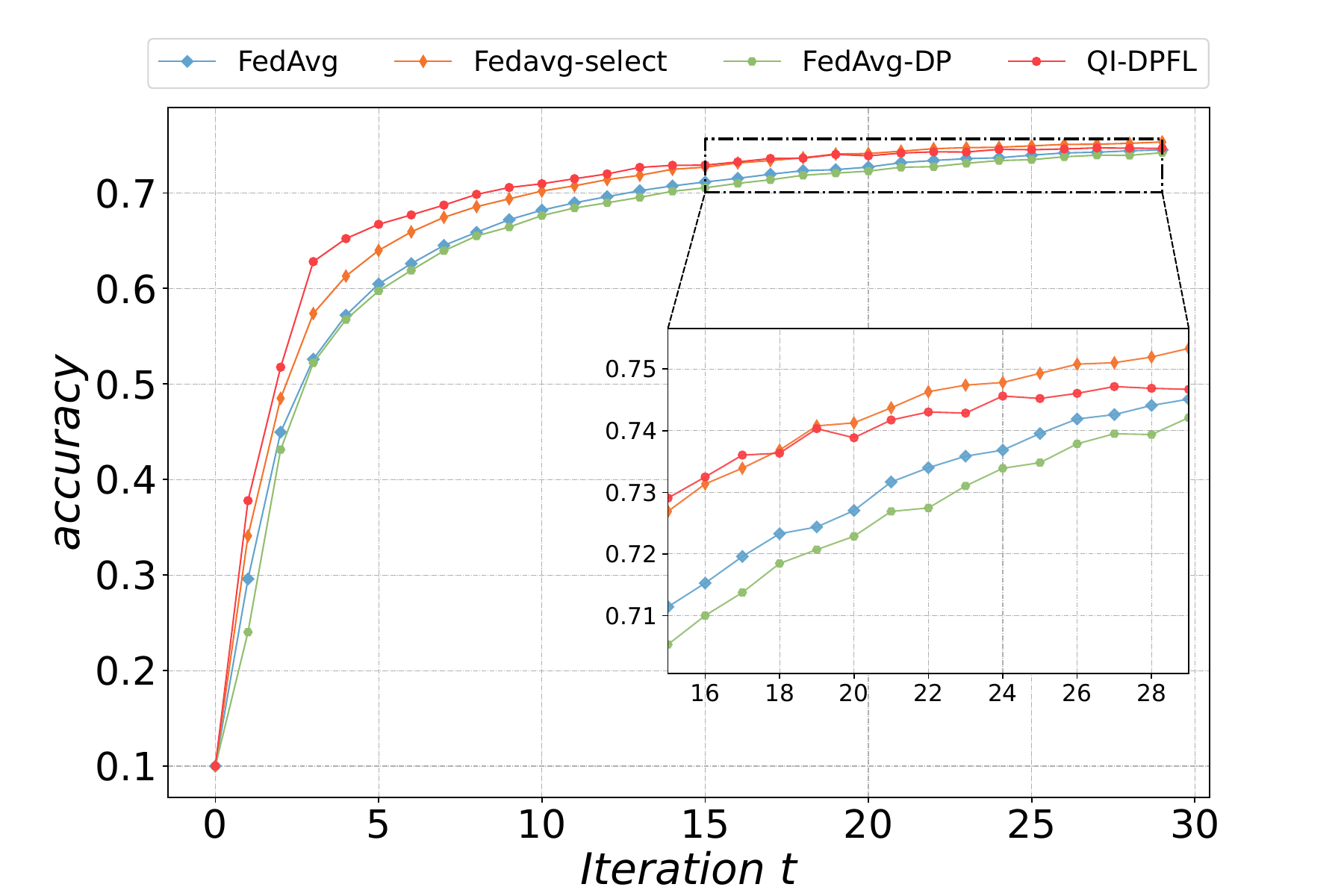}}
    \subfloat[Accuracy on different strategies]{
        \label{EMNIST_Non_IID_acc_different_strategy}
        \includegraphics[width=0.249\textwidth, trim=35 0 75 20,clip]{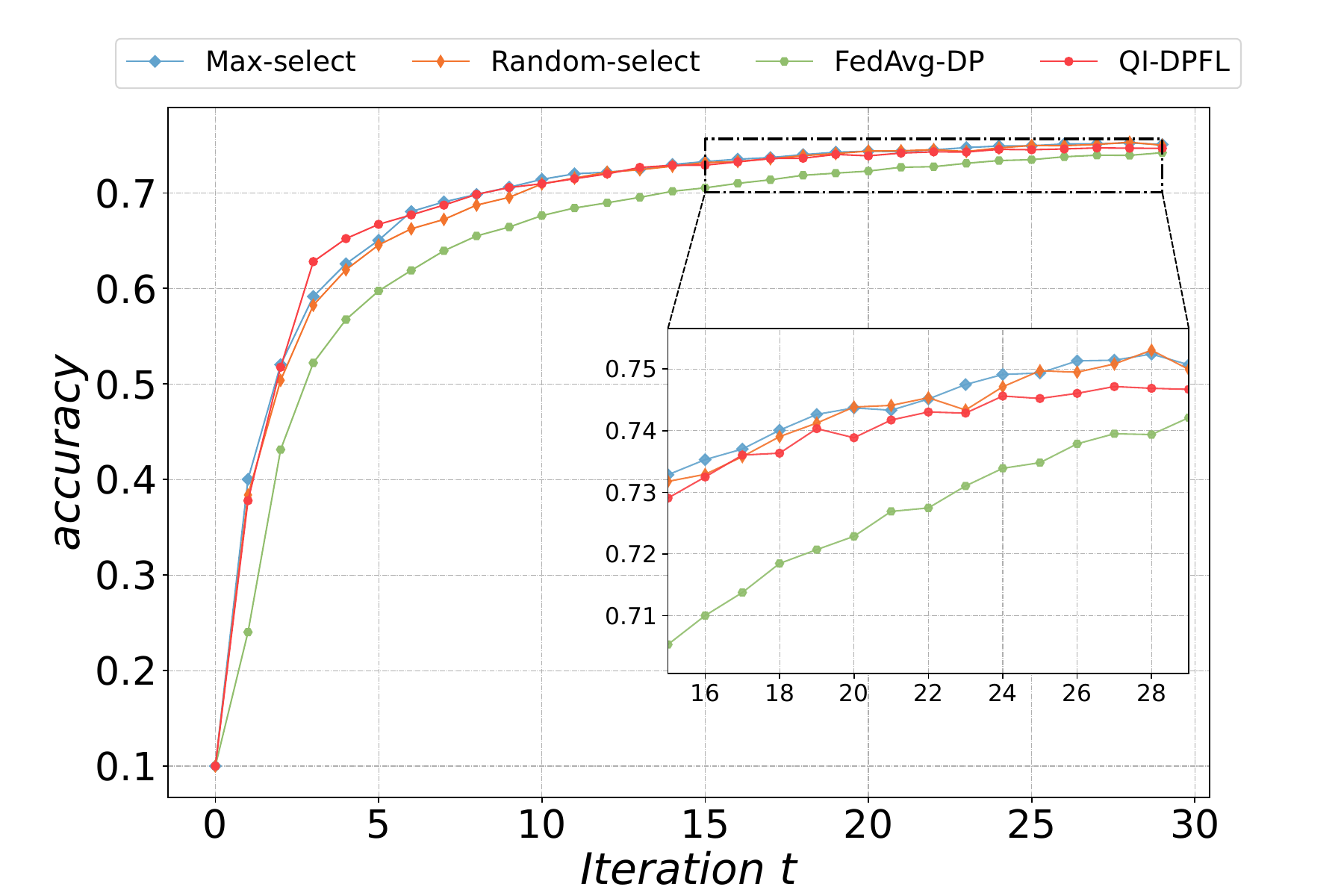}}
    \subfloat[Cost of central server]{
        \label{EMNIST_Non_IID_cost}
        \includegraphics[width=0.249\textwidth, trim=35 0 75 15,clip]{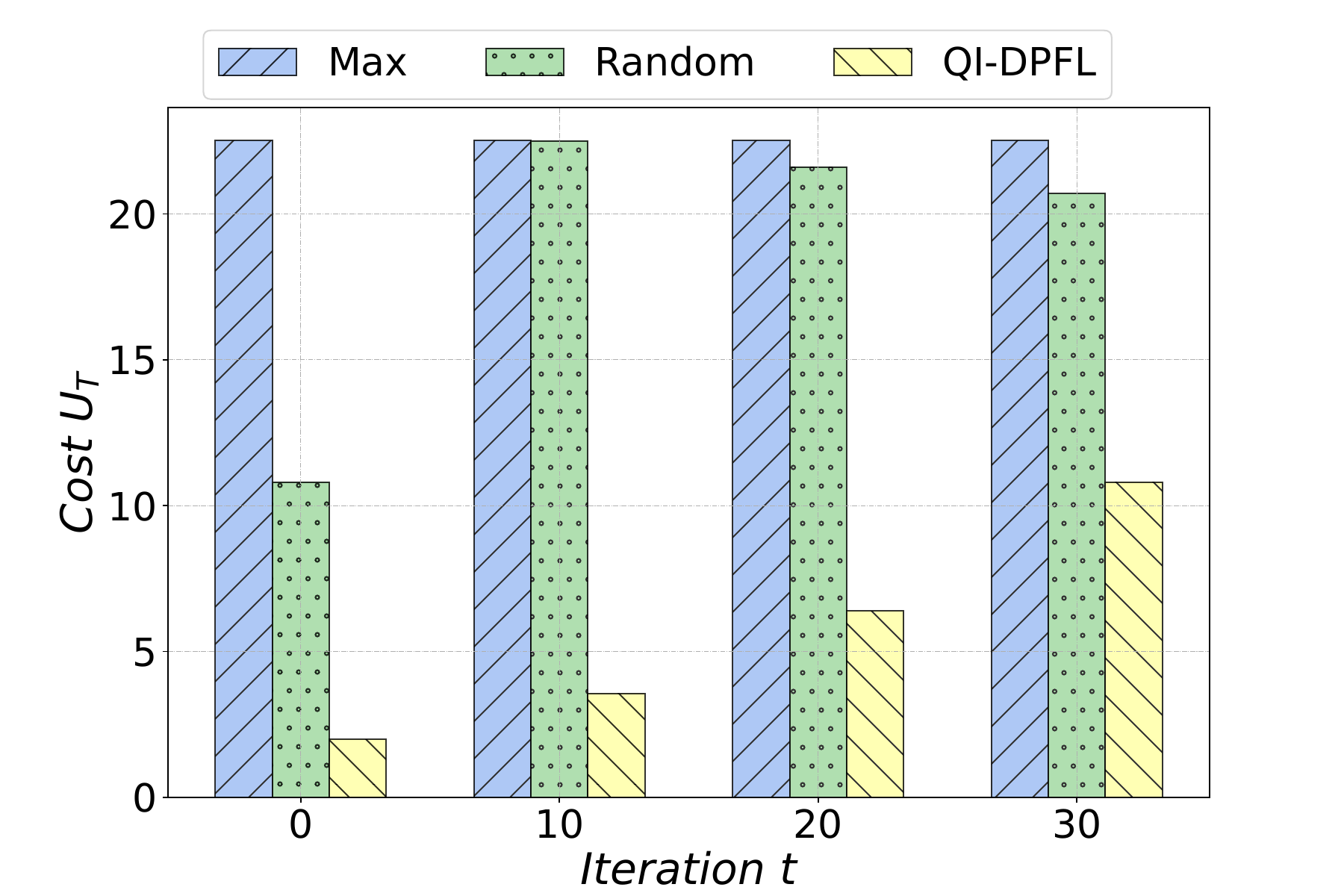}}
    \subfloat[Accuracy on different strategies]{
        \label{EMNIST_Non_IID_reward}
        \includegraphics[width=0.249\textwidth, trim=35 0 75 15,clip]{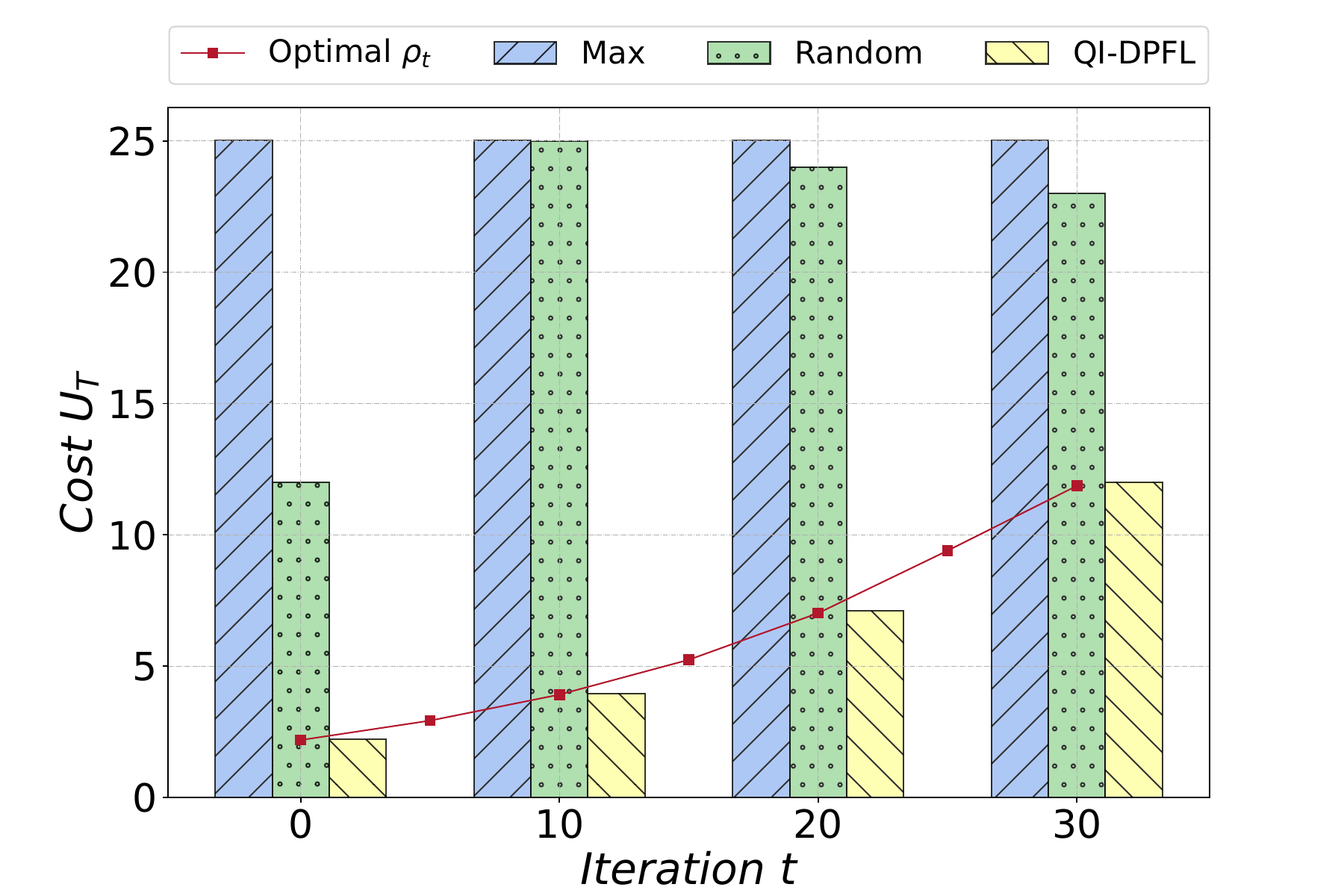}}
    \caption{The performance of different models on EMNIST dataset with Non-IID data distribution.}
    \label{EMNIST_Non_IID}
\vspace{-18pt}    
\end{figure*}

\begin{figure*}
\setlength{\abovecaptionskip}{2pt}
    \centering
    \subfloat[Accuracy on different models]{
        \label{Cifar10_IID_acc_different_model}
        \includegraphics[width=0.249\textwidth, trim=35 0 75 20,clip]{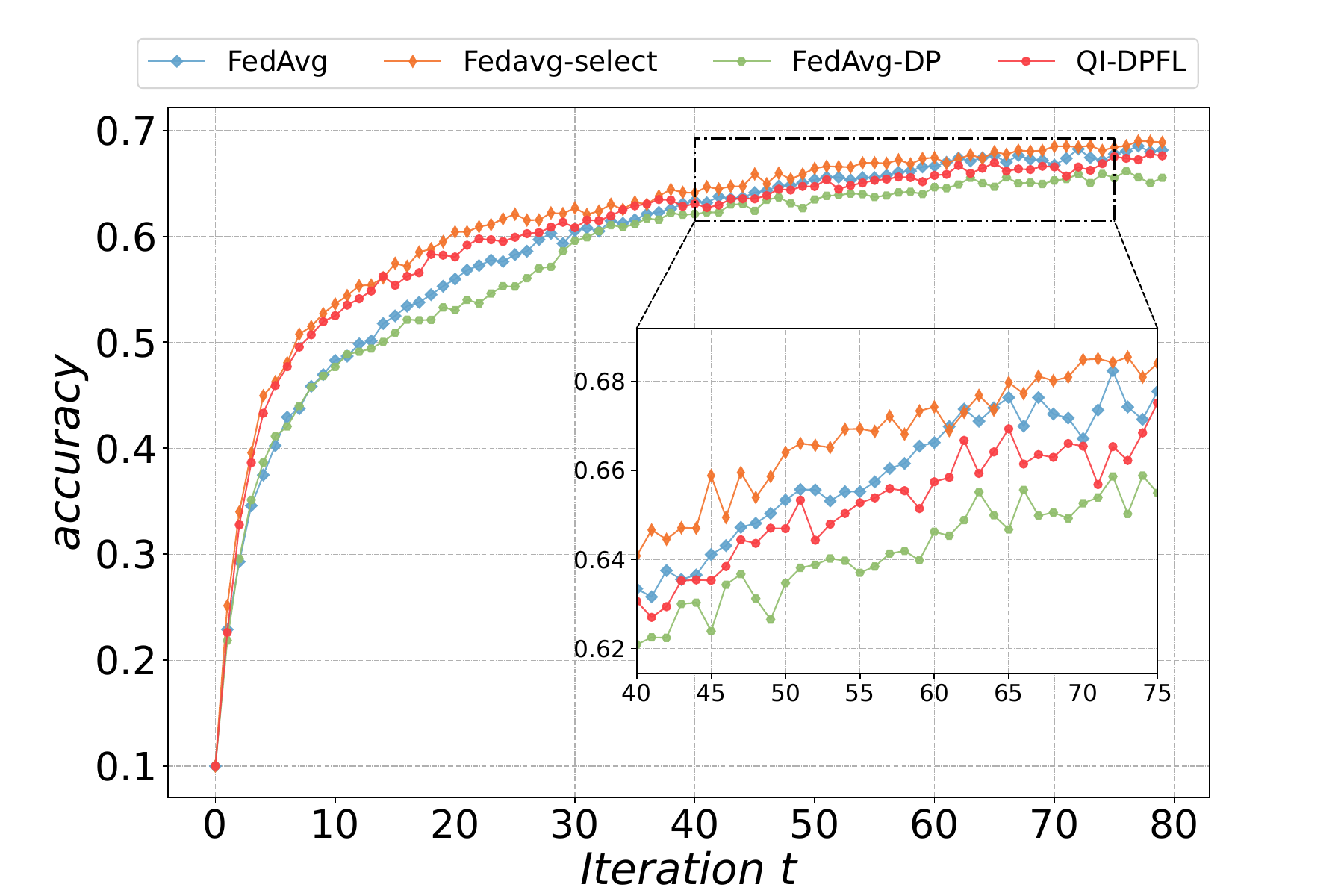}}
    \subfloat[Accuracy on different strategies]{
        \label{Cifar10_IID_acc_different_strategy}
        \includegraphics[width=0.249\textwidth, trim=35 0 75 20,clip]{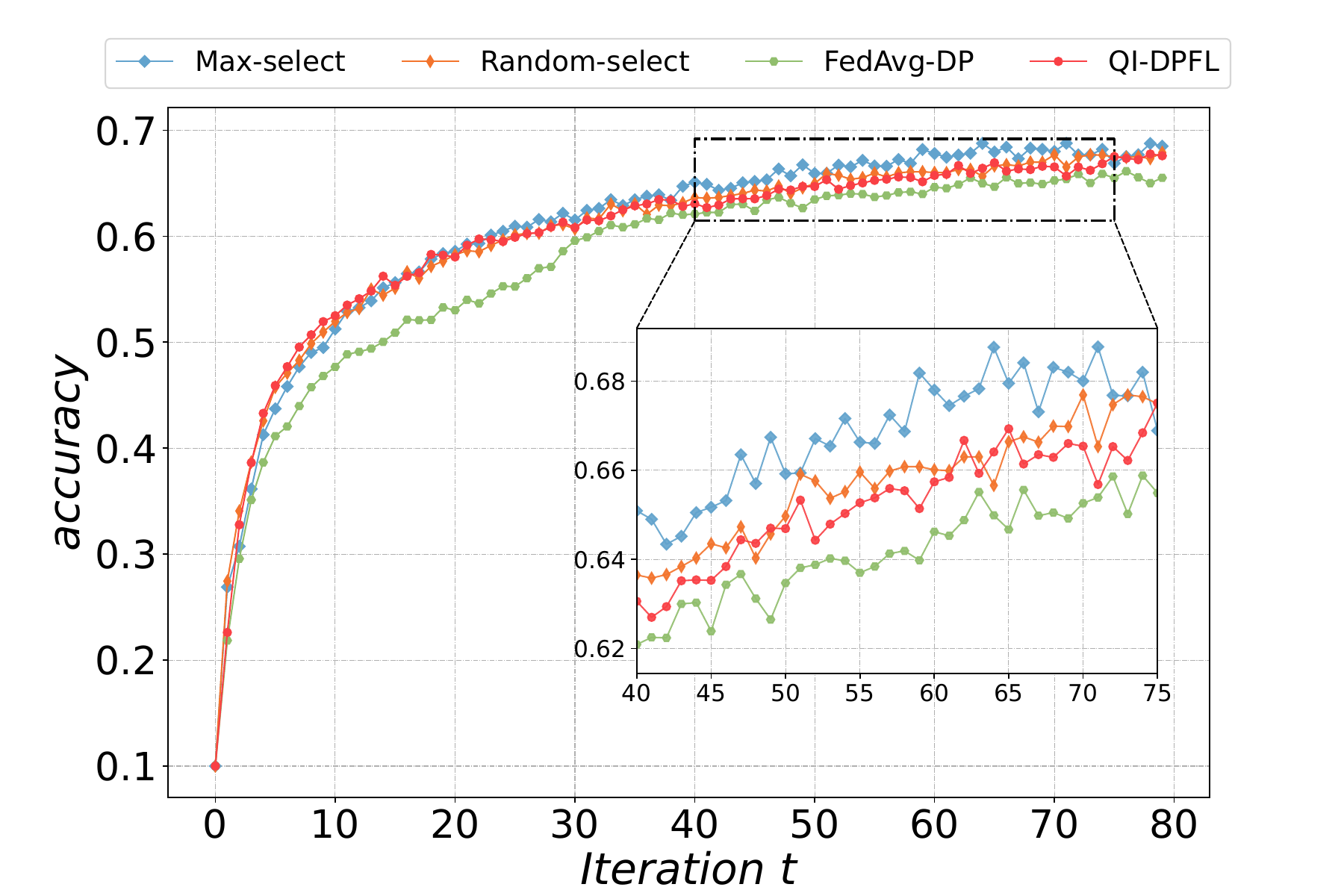}}
    \subfloat[Cost of central server]{
        \label{Cifar10_IID_cost}
        \includegraphics[width=0.249\textwidth, trim=35 0 75 15,clip]{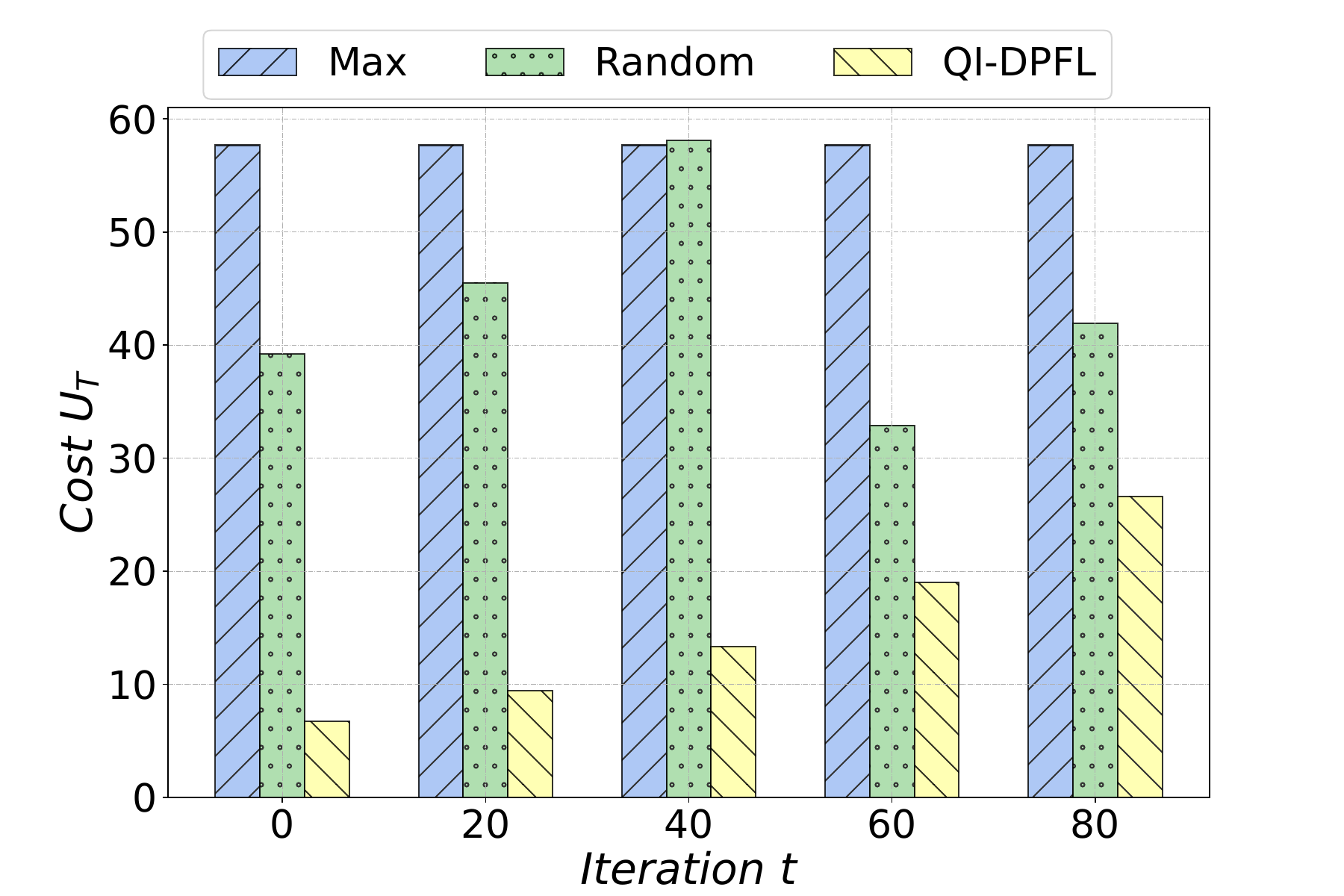}}
    \subfloat[Reward]{
        \label{Cifar10_IID_reward}
        \includegraphics[width=0.249\textwidth, trim=35 0 75 15,clip]{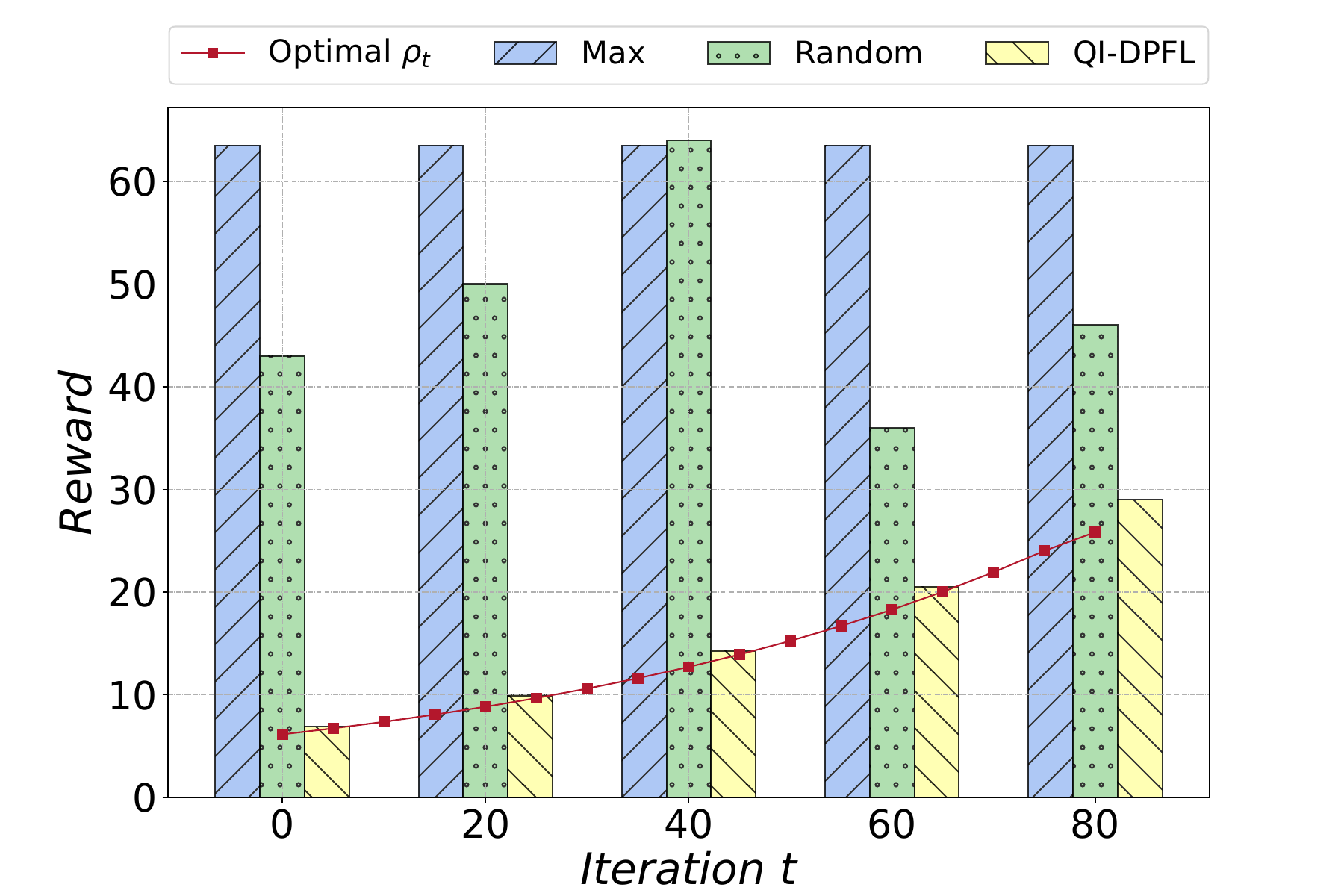}}
    \caption{The performance of different models on Cifar10 dataset with IID data distribution.}
    \label{Cifar10_IID}
\vspace{-18pt}    
\end{figure*}

\section{Numerical Experiments}
In this section, we conduct extensive experiments to demonstrate the efficiency of our proposed framework on commonly used real-world datasets in federated learning.

\subsection{Experimental Settings}
Our proposed approach is implemented on three real-world datasets (i.e., MNIST \cite{lecun1998gradient}, Cifar10 \cite{krizhevsky2009learning} and EMNIST \cite{cohen2017emnist}) to demonstrate the generalization of our framework. We adopt Dirichlet distribution \cite{hsu2019measuring} to partition the datasets into Non-IID by setting the concentration parameter as 1.0. In addition, three classic machine-learning models with different structures are implemented for each dataset. Specifically, for the MNIST dataset, we leverage a linear model with a fully connected layer of 784 input and 10 output channels. As to the Cifar10 dataset, the local training model is consistent with \cite{mcmahan2017communication}. For the EMNIST dataset, we adopt a CNN similar to the structure of LeNet-5 \cite{lecun1998gradient}. The correlated basic dataset information detailed parameter settings are summarized in Table \ref{experiment_information}.

\subsection{Experiments on Real Datasets}
In the experiments, we denote the approach without client selection and DP mechanism as \textbf{FedAvg}, the approach only performs client selection as \textbf{FedAvg-select}, the approach only considers DP mechanism as \textbf{FedAvg-DP}, and our proposed method with both client selection and DP mechanism as \textbf{QI-DPFL}. Moreover, to verify the efficiency and effectiveness of our proposed FL framework, we compare it with two baselines, named $\textbf{Max}$ and $\textbf{Random}$, which is similar to the comparison paradigm in \cite{kang2023incentive}. Specifically, $\textbf{Max}$ means the central server chooses the largest possible reward value in each global iteration to guarantee the best performance. $\textbf{Random}$ refers to selecting a random reward for clients incentivizing. Other incentive mechanisms such as \cite{hu2020trading} and \cite{sun2022profit} are designed based on different objectives, namely to maximize the utility of both clients and the central server, and maximize the profit of the model marketplace by designing an auction scenario for DPFL, respectively, which can't be compared directly with our QI-DPFL framework. Thus, We exclude these two closely related methods from our comparative analysis.

\noindent \textbf{EMNIST:} For IID data distribution, as shown in Fig. \ref{EMNIST_IID_acc_different_model}, compared with FedAvg, FedAvg-DP considers DP mechanism for privacy preservation and obtains the lowest model accuracy. As FedAvg-select and QI-DPFL select clients with high-quality data, they improve the model performance on convergence rate and accuracy. In Fig. \ref{EMNIST_IID_acc_different_strategy}-\ref{EMNIST_IID_reward}, our proposed QI-DPFL achieves the lowest cost and allocated reward in the premise of guaranteeing model performance as the Max-select method and achieves higher accuracy than FedAvg-DP. For Non-IID data distribution, the superiority of QI-DPFL is more pronounced. In Fig. \ref{EMNIST_Non_IID_acc_different_model}, FedAvg-select and QI-DPFL improve the convergence rate and model accuracy by considering the client selection mechanism compared to FedAvg. FedAvg-DP obtains the lowest accuracy as artificial Gaussian random noise is added to avoid privacy leakage. Our proposed framework QI-DPFL with perturbation on local parameters still achieves a similar model performance as FedAvg-select, which shows the effectiveness of QI-DPFL. In Fig. \ref{EMNIST_Non_IID_acc_different_strategy}-\ref{EMNIST_Non_IID_reward}, our algorithm attains model performance that is on par with the Max-select method while keeping costs and rewards minimal. Further, QI-DPFL outperforms FedAvg-DP by selecting participants with superior data quality. The experimental result on the MNIST dataset is similar to that on the EMNIST dataset.

\begin{table}[t]
\caption{Basic Information of Datasets and Parameter Settings}
\vspace{-20pt}
\begin{center}
\resizebox{0.49\textwidth}{!}{\begin{tabular}{c|c|c|c|c|c|c|c}
\toprule[1pt]

\multirow{1}{*}{\multirowcell{1}{\centering\textbf{Datasets}}} & \multicolumn{1}{c|}{\parbox{1.2cm}{\centering\textbf{Training}\\\textbf{Set Size}}} & \multicolumn{1}{c|}{\parbox{1.2cm}{\centering\textbf{Validation}\\\textbf{Set Size}}} & \multirow{1}{*}{\multirowcell{1}{\centering\textbf{Class}}} & \multicolumn{1}{c|}{\parbox{1.5cm}{\centering\textbf{Image Size}}} & $\eta$ & $T$ & \multicolumn{1}{c}{\parbox{1.2cm}{\centering\textbf{Discount}\\\textbf{Factor $\pi$}}} \\ \cmidrule[0.5pt](l{1pt}r{0pt}){1-8}

MNIST & 60,000 & 10,000  & 10 & 1 $\!\times\!$ 28 $\!\times\!$ 28 & 0.01 & 30 & 0.9429  \\ \cmidrule[0.5pt](l{1pt}r{0pt}){1-8}

Cifar10 & 50,000 & 10,000  & 10 & 3 $\!\times\!$ 32 $\!\times\!$ 32 & 0.1 & 80 & 0.9664 \\ \cmidrule[0.5pt](l{1pt}r{0pt}){1-8}

EMNIST & 731,668 & 82,587  & 62 & 1 $\!\times\!$ 28 $\!\times\!$ 28 & 0.01 & 30 & 0.901 \\ 

\bottomrule[1pt]
\end{tabular}}
\label{experiment_information}
\end{center}
\vspace{-25pt}
\end{table}

\begin{figure*}
\setlength{\abovecaptionskip}{2pt} 
    \centering
    \subfloat[Accuracy on different models]{
        \label{Cifar10_Non_IID_acc_different_model}
        \includegraphics[width=0.249\textwidth, trim=35 0 75 20,clip]{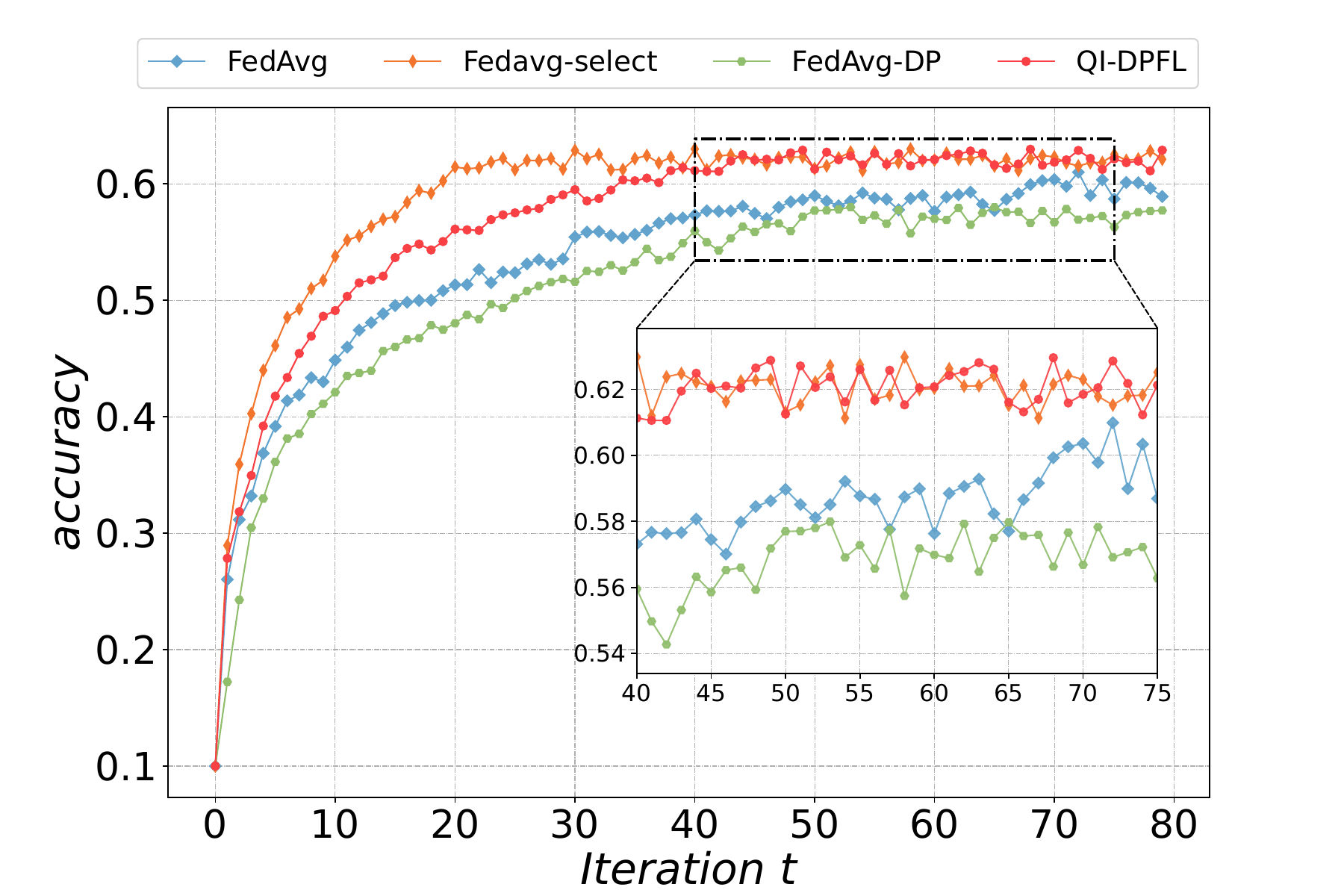}}
    \subfloat[Accuracy on different strategies]{
        \label{Cifar10_Non_IID_acc_different_strategy}
        \includegraphics[width=0.249\textwidth, trim=35 0 75 20,clip]{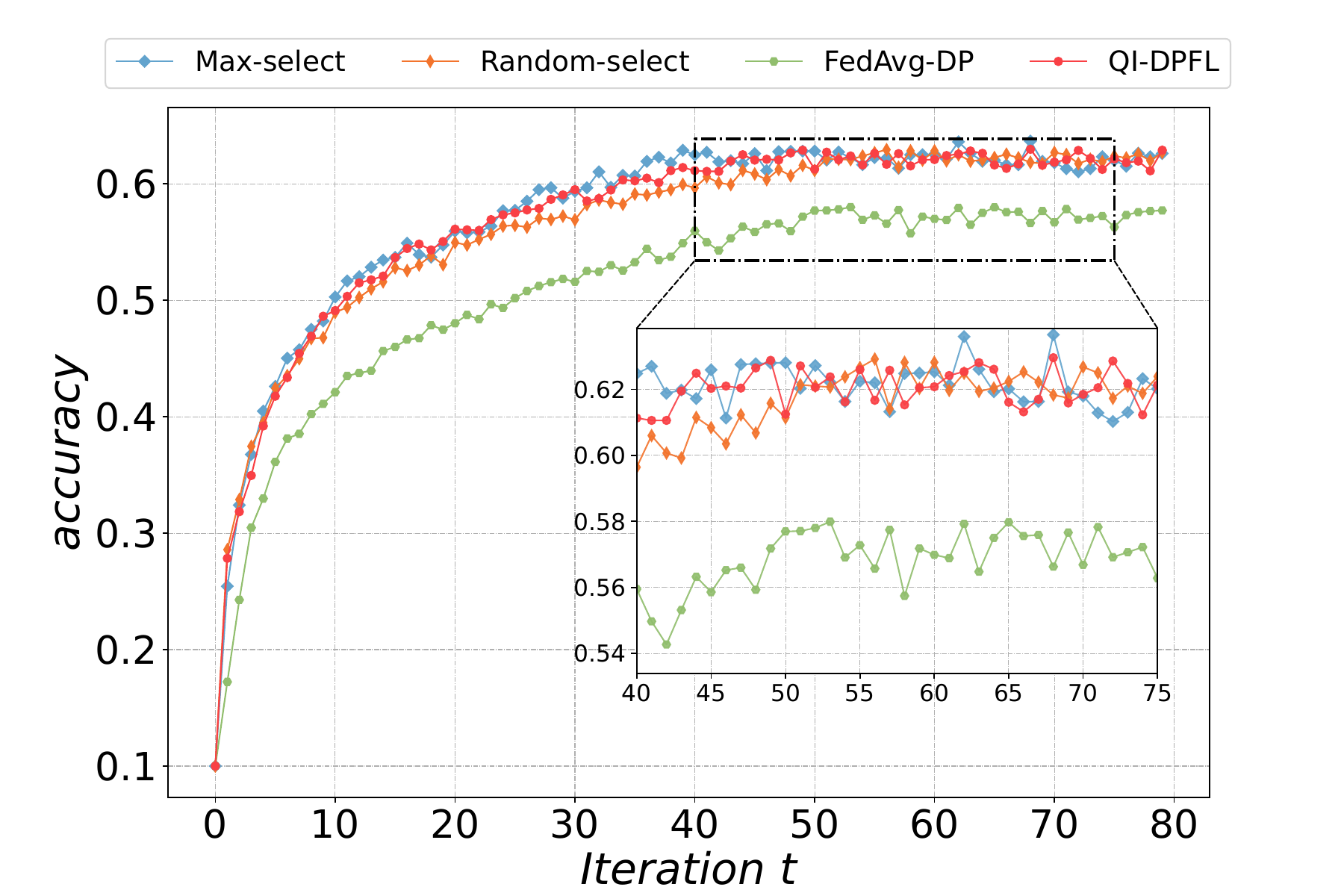}}
    \subfloat[Cost of central server]{
        \label{Cifar10_Non_IID_cost}
        \includegraphics[width=0.249\textwidth, trim=35 0 75 15,clip]{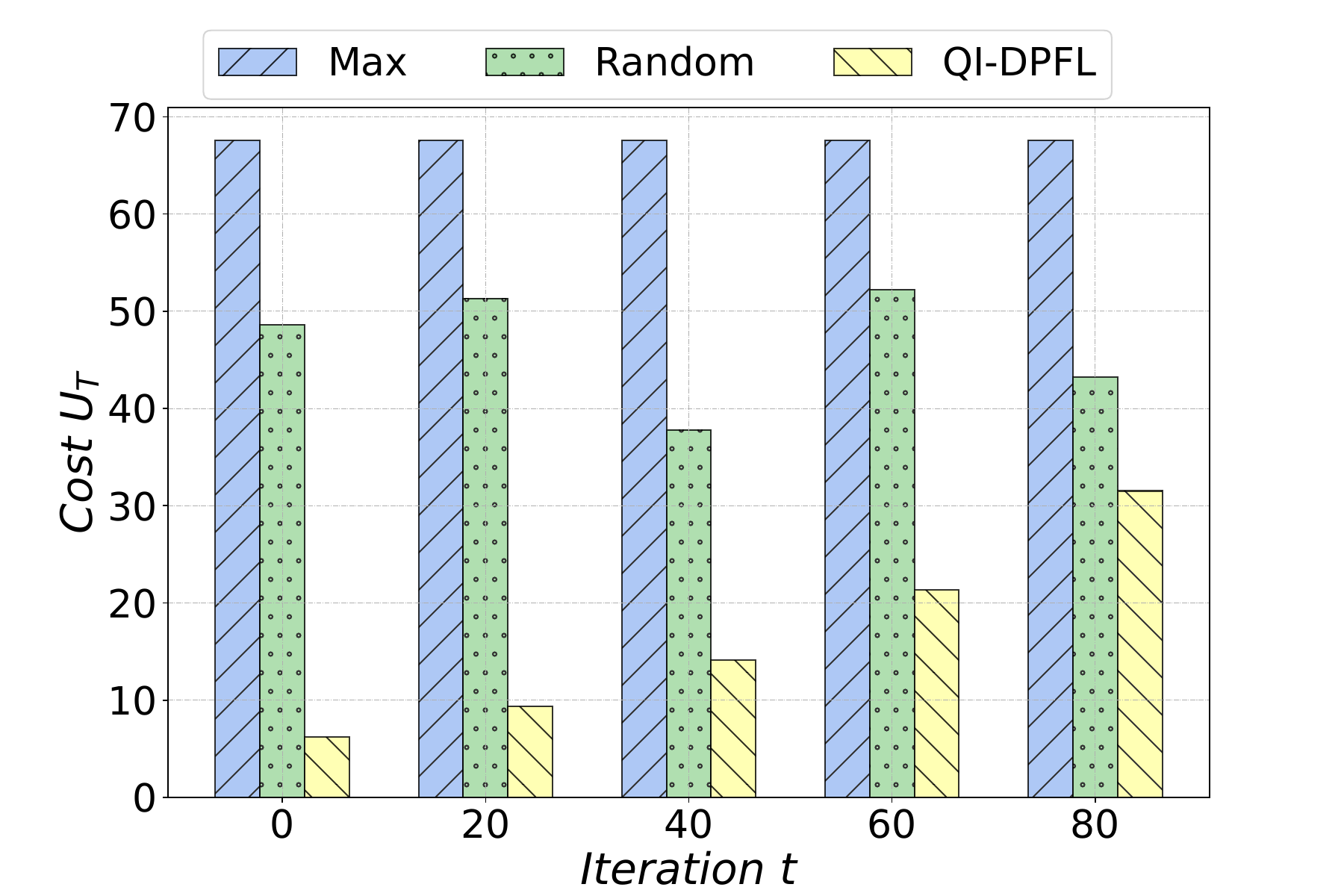}}
    \subfloat[Accuracy on different strategies]{
        \label{Cifar10_Non_IID_reward}
        \includegraphics[width=0.249\textwidth, trim=35 0 75 15,clip]{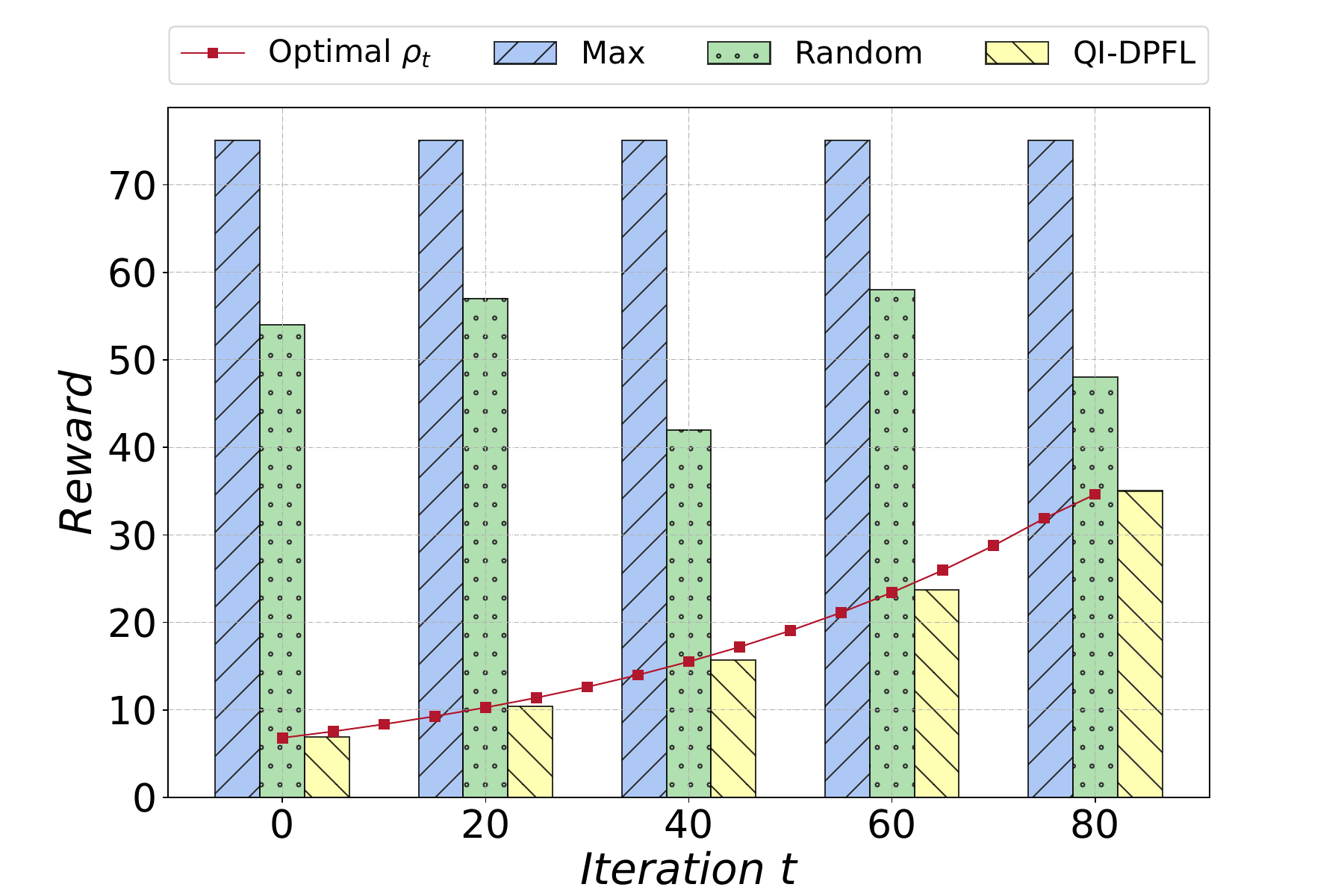}}
    \caption{The performance of different models on Cifar10 dataset with Non-IID data distribution.}
    \label{Cifar10_Non_IID}
\vspace{-20pt}    
\end{figure*}

\noindent \textbf{Cifar10:} For the IID datasets, as shown in Fig. \ref{Cifar10_IID_acc_different_model}, compared with FedAvg, the accuracy of FedAvg-select achieves a higher accuracy with the client selection mechanism. To avoid privacy leakage, the artificial noise added to the transmitted parameters in the DP mechanism may deteriorate the model performance. The accuracy of FedAvg-DP is lower than that of FedAvg as shown in Fig. \ref{Cifar10_IID_acc_different_model}. Our proposed algorithm QI-DPFL with the perturbation on transmitted local parameters still achieves a fast convergence rate and similar testing accuracy as FedAvg-select, revealing the effectiveness of QI-DPFL. Fig. \ref{Cifar10_IID_acc_different_strategy}-\ref{Cifar10_IID_reward} indicate that our approach achieves comparable model performance compared to the Max-select method while maintaining the lowest cost and reward. Additionally, QI-DPFL achieves higher accuracy than FedAvg-DP by selecting clients with high data quality. Under Non-IID data distribution, the advantages of QI-DPFL are even more conspicuous. In Fig. \ref{Cifar10_Non_IID_acc_different_model}, FedAvg-DP obtains the lowest model accuracy as the DP mechanism. FedAvg-select and QI-DPFL select clients with high-quality data, they thus exhibit enhanced performance concerning both convergence rate and model accuracy. Moreover, although QI-DPFL considers artificial noise, it still achieves similar accuracy as FedAvg-select, which demonstrates the effectiveness of our model. In Fig. \ref{Cifar10_Non_IID_acc_different_strategy}-\ref{Cifar10_Non_IID_reward}, QI-DPFL attains model accuracy that is on par with the Max-select scheme while keeping costs and rewards minimal. Further, QI-DPFL outperforms FedAvg-DP by selecting clients with superior data quality.

Based on the above experimental results, it becomes evident that the advantages of our proposed QI-DPFL are particularly pronounced on the Non-IID datasets, which is attributed to the heightened effectiveness of the client selection mechanism.

\section{Conclusion}
In this paper, we propose a novel federated learning framework called QI-DPFL to jointly solve the client selection and incentive mechanism problem on the premise of preserving clients' data privacy. We adopt Earth Mover’s Distance (EMD) metric to select clients with high-quality data. Furthermore, we model the interactions between clients and the central server as a two-stage Stackelberg game and derive the Stackelberg Nash Equilibrium to describe the steady state of the system. Extensive experiments on MNIST, Cifar10, and EMNIST datasets for both IID and Non-IID settings demonstrate that QI-DPFL achieves comparable model accuracy and faster convergence rate with the lowest cost and reward for the central server.

\bibliographystyle{IEEEtran}
\bibliography{references}

\end{document}